\documentclass[reqno]{amsart}

\usepackage{epsf,amssymb,graphics,graphicx,wrapfig}

\newcommand{\floor}[1]{\left\lfloor #1 \right\rfloor}
\def\eps{\varepsilon}

\def\N{\bbN}
\newcommand{\epsft}[1]{{\widehat{#1}^\eps}}
\def\eps{\varepsilon}

\newcommand{\Or}{{\mathcal O}}
\newcommand{\R}{{\mathbb R}}

\newcommand{\sx}{\boldsymbol\sigma_{\rm x}}
\newcommand{\sy}{\boldsymbol\sigma_{\rm y}}
\newcommand{\sz}{\boldsymbol\sigma_{\rm z}}
\newcommand{\bsone} {\boldsymbol 1}

\newcommand{\sgn}{\operatorname{sgn}}

\def\Re{{\operatorname{Re\,}}}
\def\Im{{\operatorname{Im\,}}}

\def\be{\begin{equation}}
\def\ee{\end{equation}}
\def\ba{\begin{align}}
\def\bm{\begin{multline}}
\def\bfig{\begin{figure}[htb]}
\def\efig{\end{figure}}

\numberwithin{equation}{section}
\newtheorem{theorem}{Theorem}[section]

\newtheorem{lemma}[theorem]{Lemma}

\DeclareMathSymbol{\leqslant}{\mathalpha}{AMSa}{"36}
\DeclareMathSymbol{\geqslant}{\mathalpha}{AMSa}{"3E}
\DeclareMathSymbol{\doteqdot}{\mathalpha}{AMSa}{"2B}
\DeclareMathSymbol{\circlearrowright}{\mathalpha}{AMSa}{"08}
\DeclareMathSymbol{\subsetneq}{\mathalpha}{AMSb}{"28}
\DeclareMathSymbol{\supsetneq}{\mathalpha}{AMSb}{"29}
\renewcommand{\leq}{\;\leqslant\;}
\renewcommand{\geq}{\;\geqslant\;}

\newcommand{\dd}{{\rm d}}
\newcommand{\e}[1]{\,{\rm e}^{#1}\,}
\newcommand{\ii}{{\rm i}}

\def\Re{{\operatorname{Re\,}}}
\def\Im{{\operatorname{Im\,}}}

\newcommand{\upchi}{\raise 2pt \hbox{$\chi$}}
\newcommand{\smallupchi}{\raise 1pt \hbox{\tiny $\chi$}}

\def\writefig#1 #2 #3 {\rlap{\kern #1 truecm \raise #2 truecm
\hbox{#3}}}

\newcommand{\bbC}{{\mathbb C}}\newcommand{\bbN}{{\mathbb N}}\newcommand{\bbR}{{\mathbb R}}

\newcommand{\bsI}{{\boldsymbol I}}

\begin{document}

\title[Born-Oppenheimer transitions]{Non-adiabatic transitions through tilted avoided crossings}

\author{Volker Betz and  Benjamin D. Goddard}
\address{Volker Betz \hfill\newline
\indent Department of Mathematics \hfill\newline
\indent University of Warwick \hfill\newline
\indent Coventry, CV4 7AL, England \hfill\newline
{\small\rm\indent http://www.maths.warwick.ac.uk/$\sim$betz/}
}
\email{v.m.betz@warwick.ac.uk}

\address{Benjamin Goddard \hfill\newline
\indent Department of Mathematics \hfill\newline
\indent University of Warwick \hfill\newline
\indent Coventry, CV4 7AL, England \hfill\newline
{\small\rm\indent http://www.warwick.ac.uk/staff/B.D.Goddard/}
}
\email{b.d.goddard@warwick.ac.uk}

\maketitle

\begin{quote}
{\small
{\bf Abstract.}
We investigate the transition of a quantum wave-packet through a one-dimensional avoided crossing
of molecular energy levels when the energy levels at the crossing point are tilted. 
Using superadiabatic representations, and an approximation of the dynamics near the crossing region,
we obtain an explicit formula for the transition wavefunction. Our results agree extremely well with high 
precision ab-initio calculations. 
}  

\vspace{1mm}
\noindent
{\footnotesize {\it Keywords:} Non-adiabatic transitions, superadiabatic representations, asymptotic analysis, quantum dynamics, avoided crossings}

\vspace{1mm}
\noindent
{\footnotesize {\it 2000 Math.\ Subj.\ Class.:} 81V55, 34E20}
\end{quote}

\section{Introduction}

The photo-dissociation of diatomic molecules is one of the paradigmatic chemical reactions of 
quantum chemistry. The basic mechanism is that a short laser pulse lifts the electronic 
configuration of the molecule into an excited energy state. The nuclei then feel a force due to the 
changed configuration of the electrons, and start to move according to the classical Born-Oppenheimer dynamics.
Then, at some point in configuration space, the Born-Oppenheimer surfaces of the electronic ground state and first 
excited state come close to each other, leading to a partial breakdown of the Born-Oppenheimer approximation. 
As a result, with a certain small probability the electrons fall back into the ground state, facilitating the 
dissociation of the molecule into its atoms. This important mechanism is at the heart of many processes in 
nature, such as the photo-dissociation of ozone, or the reception of light in the retina \cite{Scho94}. 
For further details on the general mechanism we refer to \cite{RRZ89}.

The mathematical problem associated with photo-dissociation are non-adiabatic transitions at avoided crossings in a 
two-level system, with one effective spatial degree of freedom.  Thus, we study the system of partial differential equations
\be \label{SE0}
\ii \eps \partial_t \psi = H \psi,
\ee
with $\psi \in L^2(\bbR, \bbC^2)$, and
\[
H = - \tfrac{\eps^2}{2} \partial_x^2 \bsI + V(x).
\]
Above,  $\bsI$ is the $2 \times 2$ unit matrix, and, with $\sigma_x$ and $\sigma_z$ the Pauli matrices as defined in (\ref{Pauli}),
\[
V(x) = X(x) \sigma_{\rm x} + Z(x) \sigma_{\rm z} + d(x) \bsI = \left( \begin{matrix} Z(x) & X(x) \\ X(x) & -Z(x) \end{matrix} \right) + d(x) \bsI
\]
is the real-symmetric potential energy matrix in the diabatic representation. Units are such that
$\hbar = 1$ and the electron mass $m_{\rm el}=1$. $\eps^2$ is the ratio of electron and
reduced nuclear mass, typically of the order $10^{-4}$. 
The timescale is such that the nuclei (with position coordinate $x$) 
move a distance of order one within a time of order one.
The motivation of 
\eqref{SE0} and its relevance for photodissociation is discussed further in \cite{BGT}. 

There is a natural coordinate transformation of \eqref{SE0} that exploits the scale separation provided by the 
small parameter $\eps$. The corresponding representation is called the adiabatic representation, and is given as follows:
Let $U_0(x)$ diagonalize $V(x)$ for each $x$, and define
$\psi_0(x) = (U_0 \psi) (x) \equiv \psi(U_0(x))$. Then $\psi_0$ solves
\be \label{SE}
\ii \eps \partial_t \psi_0 = H_0 \psi_0,
\ee
with $H_0$ given to leading order by
\be \label{adiab hamil}
H_0 = - \frac{\eps^2}{2} \partial_x^2 \bsI + \left( \begin{matrix} \rho(x) + d(x) &
- \eps \kappa_1(x) (\eps \partial_x) \\ \eps
\kappa_1(x) (\eps \partial_x) & -\rho(x)+ d(x) \end{matrix} \right).
\ee
Here, $\rho = \sqrt{X^2+Z^2}$ is half the energy level separation, and
\[
	\kappa_1 = \frac{Z'X-X'Z}{Z^2+X^2}
\] 
is the adiabatic coupling element. 
A consequence of the choice of time scale in \eqref{SE} is that solutions will 
oscillate with frequency $1/\eps$. Thus the operator $\eps \partial_x$ is actually of order one.
However, we have still achieved a decoupling of the two energy levels in \eqref{adiab hamil}, up to errors of order 
$\eps$, as long as $X^2(x) + Z^2(x) > 0$. Generically, this inequality is always true: 
assuming that the entries of $V$ are analytic in the nuclear coordinate $x$, then eigenvalues of $V$ do not cross \cite
{vNW29}, and so their difference $2 \sqrt{X^2+Z^2}$ remains positive. 
An avoided crossing is a (local or global) minimum of $\rho(x)$, which results in nonadiabatic transitions between 
the adiabatic energy levels. 

The problem of photodissociation, or more generally of non-radiative decay, can now be formulated mathematically: 
assume that \eqref{SE} is solved with an initial
wave packet $\psi_{\rm in} \in L^2(\bbR,\bbC^2)$ that is fully in the upper adiabatic level (i.e. the second
component of $\psi_{\rm in}$ is zero). This is the situation just after the laser pulse brings the electrons to their
excited state. Assuming that the initial momentum is such that the wave packet travels past an avoided crossing, we 
wish to describe the second component of $\psi_0(x,t)$, to leading order, long after the avoided crossing has been passed.
By doing this, we predict not only the probability of a molecule dissociating, but also the quantum mechanical properties 
(momentum and position distribution) of the resulting wave packet. 

The difficulty in solving the above problem is that the resulting wave packets are typically very small, namely
exponentially small in $\eps$. As an example, let us assume that the initial wave packet $\psi_{\rm in}$ has 
$L^2$-norm of order one, and that the parameters are such that the $L^2$-norm of the transmitted wave function 
is expected to be of order $10^{-6}$, which we will later see is a fairly typical value. This means that 
any straightforward numerical method with an overall error of more than $10^{-6}$ will produce meaningless results, and thus if 
we were to apply a standard method (like Strang splitting) on the full equation \eqref{SE}, we would have to use
ridiculously small time steps. To make things worse, the solution is highly oscillatory. Thus, even though 
\eqref{SE} is a system of just $1+1$ dimensional PDE's, it is not at all trivial to solve numerically. Efficient 
numerical methods to solve \eqref{SE} will therefore require insight into the analytical structure of the equation. 
 
In \cite{BG}, we used superadiabatic representations in order to obtain such insight. We derive a closed-form approximation 
to the transmitted wavefunction at the transition point, which is highly 
accurate for general potential surfaces and initial wavepackets whenever $d(x)$, the trace of the potential, is small,
but deteriorates when $d(x)$ is moderate or large at the transition point. In general, it can not be taken for granted 
in real world problems that $d(x)$ is small. Therefore, in this paper we treat a potential with an arbitrary trace. 
Our result is weaker than the one in \cite{BG}. While in the latter paper, we could allow arbitrary incoming wave functions 
as long as they were semiclassical, we essentially require the incoming wave packet to be either Gaussian or a generalized
Hagedorn wave packet in the present work. However, in that case we still obtain a closed form expression 
for the transmitted wave function at the transition point, and the accuracy is as good as in \cite{BG}. 


The importance on nonadiabatic transitions has resulted in much effort to understand them.  A simplification of the problem 
is to replace the nuclear degree of freedom by a classical trajectory.  This approach is both long-established \cite{Zen32} 
and well-understood \cite{BeLi93, HaJo04, BT05-2} and leads to the well-known Landau-Zener formula for the transition probability 
between the electronic levels.  This formula underpins a range of surface hopping models \cite{Vor98,Las07,Las08}.  
Although these and other trajectory-based methods \cite{Tul90} yield reasonably accurate transition probabilities, they are 
unable to accurately predict the shape of the transmitted wave packet  \cite{Las08}.  An improvement to the Landau-Zener 
rates is achieved by Zhu-Nakamura theory \cite{Nak02}, which is based on the full quantum scattering theory of the problem. 
However, once again only the transition probabilities are treated, and not the wave packet itself.

It is worth noting that, due to the complexity of the full quantum-mechanical problem of transitions at avoided crossings, 
there are few existing mathematical approaches.  The most relevant approach to this work is that of \cite{HaJo05} where 
another formula is given (and proved) for the asymptotic shape of a non-adiabatic wavefunction in the scattering regime at 
an avoided crossing.  For simplicity we do not state it here, see Theorem 5.1 of \cite{HaJo05}.  For comparison, their 
result looks very different to ours, and it is expected that they will not agree in the limit of small $\eps$ as theirs is 
asymptotically correct, whereas we have aimed for a simple formula which works well for a wide range of physically relevant 
parameters. Nevertheless, our approach is much better suited 
for practical purposes than the formula of \cite{HaJo05}, which requires one to calculate complex contour integrals of the 
analytic continuation of some function $V$ that is defined only implicitly. 

\section{Computing the non-adiabatic transitions} \label{S:algorithm}
In this section we will give a concise overview of our method for computing non-adiabatic transition wavefunctions, 
and explain the various parameters entering the final formula. The justification of our method, some extensions and a 
numerical test  will be given in the remainder of the paper. 

The data of our problem consists of two parts, the potential energy matrix $V$ and the initial wave function. 
More precisely, we assume that we are given $\rho(x)$ and $d(x)$ as in \eqref{adiab hamil}, and that $\rho$ has a 
unique global minimum in the region of space that we are interested in. We choose the coordinate system such that
this minimum occurs at $x=0$, and we thus have 
\[
\rho(x) = \delta + \Or(x^2), \qquad d(x) = d_0 + \lambda x + \Or(x^2).
\]
The transmitted wavefunction only depends on on $\lambda$ and $\rho$, but unfortunately the latter quantity does 
not enter in a simple way. Under the reasonable assumption that the matrix elements $X$ and $Z$ are analytic functions of $x$ 
at least close to the real axis, then so is $\rho^2$. We 
write $\rho(q)^2=\delta^2 + g(q)^2$ where $g$ is analytic and $g(0)=0$. 
Since $g^2$ is quadratic at 0, a Stokes line (a curve with $\rm{Im}(\rho)=0$) crosses the real axis perpendicularly, and, for small $\delta$, 
extends into the complex plane to two complex zeros of $\rho$, namely $q_\delta$ and $q_\delta^*$.  We define, for any complex $z$, the 
`natural scale' \cite{BeLi93} 
\[
\tau(z) = 2 \int_0^z \rho(\xi) \, \dd \xi,
\]
and write $\tau_\delta = \tau(q_\delta)$, where $q_\delta$ by convention is the complex zero with positive imaginary part. We write 
\[
\tau_{\rm r} = \Re(\tau_\delta), \qquad \tau_{\rm c} = \Im(\tau_\delta),
\]
which are the two parameters that enter into the transition formula. When we are given $\rho$ in a functional form, neither the computation
of its complex zeroes not of the complex line integral leading to $\tau_{\delta}$ is a problem numerically, and can be carried out to 
any required accuracy. However, in the case of radiationless transitions, the potential energy surfaces are often known only approximately. 
As our final formula will depend very sensitively on the value of $\tau_\delta$, small errors in this quantity will lead to wrong predictions. 
This is not a fault of our method, but a general obstruction to any numerical method aiming to calculate small nonadiabatic transitions: namely, 
since our formula below agrees very accurately with ab-initio computations, and depends very sensitively on $\tau_\delta$, getting $\rho$ wrong 
will lead to wrong results regardless of the method used. 
In a way, it should not be too surprising that when looking for a very small effect, we need to get the data right 
with very high accuracy. But it does pose a serious practical challenge when trying to predict small nonadiabatic transitions. 

As for the initial wavefunction, first of all we assume it to be initially concentrated in the upper electronic energy band. This means that 
we will consider equation \eqref{SE} with initial condition $\psi_0(x,0) = (\phi_+(x,0),0)^T$. The restriction of this work, when compared to the 
case with $\lambda=0$ considered in \cite{BG}, is on the form of $\phi_+(x,0)$, which we require to be either Gaussian, or a finite 
linear combination of Gaussians, or a Hagedorn wavefunction. For the present exposition, we restrict to the case where it is Gaussian. 
The first step of our algorithm is straighforward:

{\bf Step 1:} 
Solve the upper band adiabatic equation 
$\ii \eps \partial_t \psi_+ = H^+ \psi_+$, $\psi_+(0) = \phi_+(\cdot,0)$,
where $H^+ = -\eps^2 \partial_x^2 / 2 + \rho(x) + d(x)$. 
This can be done either by direct Strang splitting, or using the theory of Hagedorn wave packets \cite{Lubich08}. 
For a transition to occur, we need the wave packet to cross the transition region near $x=0$, where $\rho$ is minimal. So we monitor the 
expected position $\langle X \rangle = \int x |\psi_+(x)|^2 \, \dd x$ and stop the evolution when $\langle X \rangle = 0$, say at time 
$t_0$. Let us write $\phi(x) = \psi_+(x,t_0)$. $\phi$ is Gaussian up to errors of order $\eps$ \cite{Lubich08}, and centered at $x=0$. 
Thus we have 
\be \label{gauss packet}
\epsft{\phi}(k) = \exp \left( - \frac{c}{\eps} (k-p_0)^2 \right)
\ee
with parameters $p_0$ (the mean momentum) and $c$. Above, we used the semiclassical Fourier transform, cf. \eqref{epsft}. 

Given these initial data, we need to define one further derived quantity. Put $n_0 = \frac{\tau_{\rm c}}{\eps k_0}$, where 
$k_0$ is part of the solution of the pair of equations  
\be \label{n choice}
k=\sqrt{\eta^2+4\delta}, \qquad \eta=k\big(1-\tfrac{4 c \delta(\eta-p_0)}{\tau_c}\big).
\ee
Again, the numerical value of $n_0$ is easy to obtain. In what follows, we will use the abbreviation
\[
\eta^\ast = \eta^\ast(k) = \sqrt{k^2 - 4 \delta}.
\]

{\bf Step 2} Put 
\be
\label{numericsFormula}
\begin{split} 
	\epsft{\phi_-}(k)  \approx  & \frac{1}{2\sqrt{4 \alpha_{2,0}\alpha_{0,2} - \alpha_{1,1}^2}} \exp\Big[ \frac{ \alpha_{2,0}\alpha_{0,1}^2 + \alpha_{0,2}\alpha_{1,0}^2 - \alpha_{1,0}\alpha_{0,1}\alpha_{1,1} } {\alpha_{1,1}^2 - 4\alpha_{2,0}\alpha_{0,2}}\Big] \\
	& \times (\eta^*+k)  \e{-\tfrac{\tau_c}{2\delta\epsilon}|k-\eta^*|} \e{-\ii \tfrac{\tau_r}{2\delta\epsilon}(k-\eta^*)} \e{-\ii \varphi(p_0)} \epsft{\phi}(\eta^*) \chi_{k^2 > 4\delta},
\end{split}
\ee
with 
\begin{align}
	& \alpha_{1,0} = \frac{\sgn(k) \tau_{\rm c} + \ii \tau_{\rm r} - \eta^\ast n_0 \eps - 4 c \delta (\eta^\ast - p_0)}
	{2 \delta \sqrt{\eps}} + \frac{\sqrt{\eps}}{k + \eta^\ast}, \notag\\	
	& \alpha_{0,1} = -\frac{2(n_0+1) \eps^{1/2} \lambda}{k+\eta^*}, \quad \alpha_{1,1} = -\ii \eta^*+ \frac{2(n_0+1)\lambda\eps}{(k+\eta^*)^2}, \label{alphas}\\ 
	& \alpha_{2.0} = -\frac{2 \delta n_0 \eps + {\eta^\ast}^2}{8 \delta^2} - c - \frac{\eps}{2(k + \eta^\ast)^2}
	, \quad	\alpha_{0,2} = -\ii \frac{ 2\delta \lambda}{(k+\eta^*)} - \frac{2(n_0+1)\lambda^2 \eps}{(k+\eta^*)^2},\notag 
\end{align}
and 
\[
	\varphi(p_0)=-\frac{(n_0+1)^2\eps\lambda a_0 \delta}{2(n_0+1)^2\lambda^2 \eps^2 + 2\delta^2 a_0^2}
	- \frac{1}{2}\arctan \Big(\frac{a_0 \delta}{(n_0+1)\eps\lambda}\Big) + \sgn(\lambda p_0) \frac{\pi}{4},
\]
where above $a_0 = \sqrt{p_0^2 + 4 \delta} + p_0$. While formula \eqref{numericsFormula} is trival to implement on a 
computer and produces accurate results, it would of course be desirable to interpret the various terms in a physically
meaningful way. However, we have been unable to do this. 
On the other hand, everything except the factor 
$\e{-\ii \varphi(p_0)}$ is obtained by approximations and exact Gaussian computations. 
The latter factor is more tricky. 
We include it because we found a discrepancy in the phase of the transmitted wave function 
between our formula without the factor and numerical 
ab-initio calculations, which is probably due to one of our approximations below being too crude. 
This phase discrepancy is removed by essentially computing the phase in the case 
of an incoming Gaussian when the parameter $c$ diverges, i.e.\ infinitely small momentum uncertainty, 
which gives $\varphi(p_0)$, and subtracting that. While this fixes the discrepancy with the numerics, 
we do not, as yet, fully understand why it does so, and where the original inaccurate approximation has been made. 

One important property of $\varphi(p_0)$ is that it is constant in both $k$ and $x$ and hence will not affect any quantum mechanical expectation values. It would however play a role when we consider interferences. 

The final step of our algorithm is again straightforward:\\
{\bf Step 3} Solve the lower band adiabatic equation with initial condition $\phi_-$, i.e. solve 
\[
\ii \eps \partial_t \psi_- = H^- \psi_-, \qquad \psi_-(t_0) = \phi(\cdot,t_0),
\]
where $H^- = -  \eps^2 \partial_x^2 / 2 -\rho(x) + d(x)$, and $\phi_-$ is the inverse semiclassical Fourier transform of 
$\epsft \phi_-$. For times so large that $\psi_-$ has support far away from the transition region, it describes the 
transmitted wave function of equation \eqref{SE} with great accuracy, see section \ref{S:Numerics}.

\section{Evolution in the Superadiabatic Representations}

\subsection{Superadiabatic Representations}

The key idea for deriving our transition formulae is to study the evolution in a suitable superadiabatic representation. 
For a careful discussion of the theory of those representation, we refer to \cite{BGT}. Here we give only some intuition and the 
mathematical facts. The $n$-th superadiabatic representation is implemented by a unitary operator $U_n$ 
acting on $L^2(\bbR,\bbC^2)$, and its main property is that it diagonalizes the right hand side of \eqref{SE0} up to errors 
of order $\eps^{n+1}$. Thus, the adiabatic representation \eqref{adiab hamil} is the zeroth superadiabatic representation, and 
in general 
\[
	H_n = U_n^{-1} H U_{n}
	 = - \frac{\eps^2}{2} \partial_x^2 \bsI + \left( \begin{matrix} \rho(x) + d(x) &
	\eps^{n+1} K_{n+1}^+ \\ \eps^{n+1} K_{n+1}^- & -\rho(x)+ d(x) \end{matrix} \right),
\]
where $K_n^\pm$ are the $n$-th superadiabatic coupling elements. They are usually pseudo-differential operators, and so are the 
$U_n$. The useful consequence of switching to the superadiabatic representation is that now 
the evolution of the second component $\psi_n^-$ of $\psi_n=U_n \psi$, subject to $\psi_n^-(-\infty)=0$, 
is given by
\be
	\psi_{n}^-(t) = - \ii \eps^{n}   \int_{-\infty}^t  \e{-\frac{\ii}{\eps} (t-s) H^-} K_{n+1}^- \e{-
\frac{\ii}{\eps} s  H^+ } \phi \, \dd s, \label{psinEvolution}
\ee
up to relative errors of order $\eps$. Thus, provided we can control $K_n^-$, \eqref{psinEvolution} gives the transmitted wave function in 
the $n$-th superadiabatic representation to high precision. 

There are some apparent problems with this idea. Firstly, it is far from clear how we hope to control $K_n^-$. Secondly, 
the superadiabatic unitaries are in general very hard to calcualte, and as such this formulation does not allow the adiabatic
wavefunction to be easily obtained. Thirdly, we have to decide which value of $n$ we want to use. The sequence $K_n^-$ is 
expected to be asymptotic in $n$, so after initially decaying rapidly (in an appropriate sense) it will start to grow beyond all
limits when $n$ is taken to infinity. 
The second problem is resolved when we study the wavefunction in the scattering regime,
well away from the avoided crossing.  In this case, for potentials which are approximately constant, it is known that 
$U_n$ and $U_0$ agree up to small errors depending on the derivatives of the potential \cite{Teu03}, 
and (\ref{psinEvolution}) can be used to calculate the transmitted wavefunction.
For the value of $n$, in \cite{BGT} we showed for a special choice of parameters $\rho, \kappa$ that  there exists an `optimal' $n$ 
for which $\psi_n^-(t)$ builds up monotonically, corresponding to a single transition. This $n$ is given by the set of nonlinear 
equations \eqref{n choice} that we have seen in the previous section.
We expect this set of equations to hold in general, and have obtained very good results by using it here.

The problem of calculating $K_n^-$ turns out to be reducible to a set of differential recursions, which we will now give. 
The discussion follows the one in \cite{BGT} very closely, the only difference being that we now include a nonzero trace $d(x)$ in 
the Hamiltonian. All the calculations and arguments are almost the same as 
in \cite{BGT}, so we will omit them. 

We change from the spatial representation to the symbolic representation (see e.g.\ \cite{Teu03}) 
by replacing $x$ by $q \in \R$ and $\ii \eps \partial_x$ by an 
independent variable $p \in \R$, where the factor $\eps$ takes into account the semiclassical scaling.  We need to introduce some further 
notation: we rewrite the potential as
\[
	V(q)=\rho(q)\begin{pmatrix} \cos\big(\theta(q)\big) & \sin\big(\theta(q)\big) \\ \sin\big(\theta(q)\big) & \cos \big(\theta(q)\big) \end{pmatrix} + d(q) \begin{pmatrix} 1 & 0 \\ 0 & 1 \end{pmatrix},
\]
which defines $\theta(q)$.  It follows that the unitary transformation to the adiabatic representation is given by
\[
	U_0(q)=\begin{pmatrix} \cos \big( \frac{\theta(q)}{2} \big) & \sin\big( \frac{\theta(q)}{2} \big)
	 \\ \sin\big( \frac{\theta(q)}{2} \big) & -\cos\big( \frac{\theta(q)}{2} \big) \end{pmatrix}.
\]
Hence the Pauli matrices in the adiabatic representation are given by
\[
	\sx(q)=U_0(q) \sigma_{\rm x} U_0(q), \quad \sy(q)=U_0(q) \sigma_{\rm y} U_0(q),
	\quad \sz(q)=U_0(q) \sigma_{\rm z} U_0(q),
\]
where
\be
	\sigma_{\rm x} = \begin{pmatrix} 0 & 1 \\ 1 & 0 \end{pmatrix}, \quad 
	\sigma_{\rm y} = \begin{pmatrix} 0 & -\ii \\ \ii & 0 \end{pmatrix}, \quad 
	\sigma_{\rm z} = \begin{pmatrix} 1 & 0 \\ 0 & -1 \end{pmatrix},
	\label{Pauli}
\ee
and we have used that $U_0^*=U_0$. 

A direct calculation confirms, with  $\bsone$ the $2\times 2$ identity matrix:
\begin{lemma} \label{V derivatives}
We have
\[
\partial_q^n V(q) = a_n(q) \sz(q) + b_n(q) \sx(q) + c_n(q) \bsone,
\]
where $a_n(q)$, $b_n(q)$ and $c_n(q)$ are given by the recursions
\be \label{anbn}
\begin{array}{rclrcl}
a_0(q) &=& \rho(q), \quad & a_{n+1}(q) &=& a_n'(q) + \theta'(q) b_n(q) \\[2mm]
b_0(q) &=& 0, \qquad & b_{n+1}(q) &=& b_n'(q) - \theta'(q) a_n(q) \\[2mm]
c_0(q) &=& d(q) \qquad & c_{n+1}(q) &=& c_n'(q)
\end{array}
\ee
\end{lemma}

We then have the following explicit recursion for the coupling elements:
\begin{theorem} \label{h_n lowest order}
The Hamiltonian in the $n$-th superadiabatic representation is given by 
\[
	H_n(\eps,p,q) = \frac{p^2}{2} \bsone + \begin{pmatrix}
	\rho(q) + d(q)& \eps^{n+1} \kappa_{n+1}^{+} (p,q) \\ \eps^{n+1} \kappa_{n+1}^{-}(p,q) & -\rho(q) + d(q)
	\end{pmatrix} + \begin{pmatrix} \Or(\eps^2) & \Or(\eps^{n+2}) 
	\\ \Or(\eps^{n+2}) & \Or(\eps^2)
	\end{pmatrix},
\]
where
\[
	\kappa_{n+1}^{\pm}(p,q) = -2 \rho (q) (x_{n+1} (p,q) \pm y_{n+1} (p,q)).
\]

Setting  $x_n(p,q)=\sum_{m=0}^n p^{n-m} x_n^m(q)$ with similar expressions for $y_n$, $z_n$ and $w_n$, the coefficients $x_n^m$ to $w_n^m$ are determined by the 
following recursive algebraic-differential equations:
\be \label{rec start 2}
x_1^m = z_1^m = w_1^m = 0, \; m=0,1, \qquad y_1^0 = -\ii \frac{\theta'(q)}{4 \rho(q)}, \; y_1^1=0
\ee
with
\[
x_{n+1}^m = -\frac{1}{2 \rho} \left( \frac{1}{\ii} (y_n^m)' - 2 \sum_{j=1}^{\floor{m/2}} \frac{1}{(2\ii)^j}
\tbinom{n+1-m+j}{j} (b_j z_{n+1-j}^{m-2j} - a_j x_{n+1-j}^{m-2j} + c_j y_{n+1-j}^{m-2j}) \right)
\]
for $n$ odd, and
\[
\begin{split}
y_{n+1}^m = & -\frac{1}{2 \rho} \left( \frac{1}{\ii} \big( (x_n^m)' - \theta' z_n^m \big)  
 - 2 \sum_{j=1}^{\floor{m/2}} \frac{1}{(2\ii)^j}
\tbinom{n+1-m+j}{j} ( -a_j y_{n+1-j}^{m-2j} + b_j w_{n+1-j}^{m-2j} + c_j x_{n+1-j}^{m-2j}) \right),\\
0  = & \frac{1}{\ii} \big( (z_n^m)' + \theta' x_n^m \big) - 2 \sum_{j=1}^{\floor{m/2}} \frac{1}{(2\ii)^j}
\tbinom{n+1-m+j}{j} (b_j y_{n+1-j}^{m-2j} + a_j w_{n+1-j}^{m-2j} + c_j z_{n+1-j}^{m-2j}),\\
0  = & \frac{1}{\ii} (w_n^m)' - 2 \sum_{j=1}^{\floor{m/2}} \frac{1}{(2\ii)^j}
\tbinom{n+1-m+j}{j} (a_j z_{n+1-j}^{m-2j} + b_j x_{n+1-j}^{m-2j} + c_j w_{n+1-j}^{m-2j}),
\end{split}
\]
for $n$ even.  The coefficients $a_n$ to $c_n$ are given by Lemma \ref{V derivatives}
\end{theorem}
\begin{proof}
The proof is analogous to those of Theorem 3.4 and Proposition 3.5, with some easy alterations due to the presence of $d(x)$. 
\end{proof}

We note that, as in the trace-free case in \cite{BGT}, $y_n^m=0$ for all $m$ when $n$ is even and $x_n^m=z_n^m=w_n^m=0$ for all $m$ when $n$ is odd.
Furthermore, from the above equations, it is obvious that $x_n^m = y_n^m = z_n^m = w_n^m = 0$ for odd $m$.

We now have an explicit expression for $\kappa_n^-$ and may therefore also calculate $K_n^-$, the superadiabatic coupling element, which is the Weyl quantization of the symbol $\kappa_n^-$:
\[
	K_{n}^\pm \psi (x) = \frac{1}{2 \pi \eps} \int_{\R^2} \dd \xi\,  \dd y \, \kappa_{n}^{\pm}
	\left( \tfrac{x+y}{2}, \xi \right) \e{\frac{\ii}{\eps} \xi (x-y)} \psi(y),
\]
and from the recursions in Theorem \ref{h_n lowest order}, it is clear that
\[
	\kappa_n(p,q)=\sum_{j=0}^n p^j \kappa_{n,n-j}(q),
\]
where the $\kappa_{n,n-j}$ can be calculated explicitly. Determining the asymptotics of this two-parameter recursion is 
a very tricky problem to which we have no solution. However, in the regime of large $p$ (meaning, large incoming momentum)  the sum is well-
approximated by the $j=n$ term. For $p=\mathcal{O}(\eps^{-1/3})$, this can be made rigorous on the level of the superadiabatic 
Hamiltonian, while a full asymptotic investigation of the transitions in this regime is still work in progress \cite{BG2}.
Here, we use this approximation without further justification, and find that it gives good results even for relatively small
values of $p$.  

The asymptotics of the term $\kappa_{n,0}^-$ can be determined explicitly in the following generic case.  Without loss of generality, we assume 
that the avoided crossing occurs at $x=0$, specify the initial wave packet at $t=0$, and write $\rho(q)^2=\delta^2 + g(q)^2$ where $g$ is 
analytic and $g(0)=0$.    As is standard in asymptotic analysis (see e.g. \cite{BeLi93}), the asymptotic behaviour of $\kappa_{n,0}^-$ is 
determined by the complex zeros of $\rho$.  Since $g^2$ is quadratic at 0, a Stokes line (a curve with $\rm{Im}(\rho)=0$) crosses the real axis 
perpendicularly, and, for small $\delta$, extends into the complex plane to two complex zeros of $\rho$, namely $q_\delta$ and $q_\delta^*$.  As 
argued by Berry and Lim, in the natural scale $\tau(q)=2 \int_0^q \rho(q) \dd q$ and near $q=0$, the adiabatic coupling function has the form $
	\kappa_1(q) = \tfrac{\ii \rho(q)}{3}\big( \tfrac{1}{\tau(q)-\tau_\delta^*} - \tfrac{1}{\tau(q)-\tau_\delta} + \kappa_r ( \tau(q) )\big),$
with $\tau_\delta = \tau(q_\delta)$.  In particular, $\kappa_r$ has no singularities fro $|\tau| < |\tau_\delta|$, and 
no singularities of order $\geq1$ for $|\tau|\leq |\tau_\delta|$.  As can be seen from Theorem \ref{h_n lowest order}, 
solving the recursions for $\kappa_n^-$ requires taking high derivatives of $\kappa_1$.  By the Darboux principle, the 
asymptotics are dominated by the complex singularities closest to the real axis, $\tau_\delta$ and $\tau_\delta^*$.  Hence, 
to leading order, we find
\be
	\kappa_{n,0}^-(q)=\tfrac{\ii^n}{\pi} \rho(q) (n-1)! \big( \tfrac{\ii}{(\tau-\tau_\delta^*)^n} - \tfrac{\ii}{(\tau-\tau_\delta)^n} \big).
	\label{kappan0}
\ee

Using the definition of the Weyl quantisation, a direct calculation \cite{BGT} shows 
\be
	K_{n,0}^-=\sum_{j=0}^n \binom{n}{j}\big( \tfrac{\eps}{2\ii} \big)^j \big(\partial_x^j \kappa_{n,0}^-(x) \big) (-\ii \eps \partial_x)^{n-j}.
	\label{Kn0minus}
\ee

\subsection{Approximation of the Adiabatic Propagators} \label{S:ApproxProp}
In order to determine a closed form approximation for  (\ref{psinEvolution}), it is necessary to approximate the adiabatic propagators. 
This is in contrast to the situation in \cite{BGT} where the model was chosen such that $\rho$ is constant, and thus 
the adiabatic evolutions were trivial in Fourier space. 

The first insight is that the operator $K_{n,0}^-$ given in (\ref{Kn0minus}) is sharply localized:
$K_{n,0}^- f$ will only be significantly different from zero if either $f$ or some of its derivatives have some support overlap with 
$\kappa_{n,0}^-$, which means they must be concentrated near the real solution of $\Re(\tau(q)) = \Re(\tau_\delta)$ that is closest to 
$q=0$. We will refer to this solution as the  transition point. In Section \ref{S:n} we will see that 
relevant values of $n$ are of the order $1/\eps$; furthermore, for large $n$ we have $(1+x^2)^
{-n} \approx \e{-n x^2}$, and so $\kappa_{n,0}^-$ and its derivatives are concentrated in a $\sqrt\eps$ neigbourhood of the 
transition point.  Since the time scale is chosen such that the 
semiclassical wave packets (which have width of order $\sqrt\eps$) travel at speed of order one, the dominant transitions come from a time
interval of order $\sqrt\eps$ around the transition time, which we define to be the time when the expected position of the incoming 
wave packet crosses the transition point. 

Let us pick a coordinate system so that that the transition time is $s=0$. We cannot, however, choose the 
transition point to be at $x=0$, since we have already fixed $x=0$ to be the local minimum of $\rho$. On the other hand, one of our later 
calculations relies on the fact that the transition point is at least in a $\sqrt{\eps}$ neighbourhood of $0$, see Section \ref{S: FT approx}. 
So from now on, we will always assume that the transition point does indeed have this property. This assumption can be justified by the 
observation that for sensible potentials, the real and imaginary parts of the complex zeroes of $\rho$ are coupled, and are either 
both relatively small or both large. However, in the latter case, transitions tend to be so small that they are physically uninteresting. 
That said, it would of course be much preferable to be able to treat arbitrary transitions, but we cannot do this yet. In what follows,
we will always pretend that the transition point is $x=0$, although for the calculation in the next paragraph below this is not yet strictly 
necessary. 

The above considerations allow us to replace the potential in the full adiabatic dynamics by its first Taylor approximation, as the following
formal calculation shows.
We take $H_1^\pm:=-\eps^2 \partial_x^2/2 \pm \delta + 
\lambda x$ and wish to show that $\e{-\tfrac{\ii}{\eps}sH^\pm} - \e{-\tfrac{\ii}{\eps}sH_1^\pm}$ is small.
We have
\[
\begin{split}
	\e{-\tfrac{\ii}{\eps}sH^\pm} - \e{-\tfrac{\ii}{\eps}sH_1^\pm} &= \e{-\tfrac{\ii}{\eps}sH_1^\pm}
	\big(\e{\tfrac{\ii}{\eps}sH_1^\pm}\e{-\tfrac{\ii}{\eps}sH^\pm} - 1\big) = \e{-\tfrac{\ii}{\eps}sH_1^\pm}
		\int_0^s \partial_r \big(\e{\tfrac{\ii}{\eps}sH_1^\pm} \e{-\tfrac{\ii}{\eps}sH^\pm} \big) \dd r\\ 
		& = \e{-\tfrac{\ii}{\eps}sH_1^\pm}
		\int_0^s  \e{\tfrac{\ii}{\eps}sH_1^\pm}\big( \tfrac{\ii}{\eps}(H_1^\pm-H^\pm)\big) \e{-\tfrac{\ii}{\eps}sH^\pm}\dd r.
\end{split}
\]
We now note that $H_1^\pm-H^\pm$ is quadratic near $x=0$ and hence the integrand is of order 1 in a $\sqrt{\eps}$ neigbourhood of zero.  Hence the left hand side is bounded by the length of the integration region $\sqrt{\eps}$ and, to leading order, it suffices to replace  (\ref{psinEvolution}) by
\be
	\psi_{n}^-(t) \approx - \ii \eps^{n}   \e{-\tfrac{\ii}{\eps}t H^-}\int_{-\infty}^t  \e{\frac{\ii}{\eps}s (-\eps^2 \partial_x^2/2-\delta+\lambda x)} K_{n+1}^- \e{-
\frac{\ii}{\eps} s  (-\eps^2 \partial_x^2/2+\delta+\lambda x) } \phi \, \dd s, \label{psinH1}
\ee
where we have not altered the $s$-independent propagator.

We now find it convenient to switch to the Fourier representation by applying the semiclassical Fourier transform
\be
	\epsft{f}(k) = \frac{1}{\sqrt{2 \pi \eps}} \int_\R \e{-\frac{\ii}{\eps}k q} f(q) \, \dd q = \frac{1}{\sqrt{\eps}}\hat{f}\big(\tfrac{k}{\eps}\big).
	\label{epsft}
\ee
We define $\hat{K}_n$ through $\hat{K}_n \epsft{\psi} = \epsft{K_n\psi}$, and a direct calculation \cite{BGT} gives
\[
	\hat K_{n,0}^{\pm} f(k) = \frac{1}{\sqrt{2 \pi \eps}}
	\int_\bbR \dd \eta \, \epsft{\kappa_{n,0}^{\pm}} (k - \eta) \left(\tfrac{\eta+k}{2}\right)^n f(\eta).
\]
Fourier transforming both sides of (\ref{psinH1}), we see that $\epsft{\psi_n^-}$ is given by a double integral:
\[
	\epsft{\psi_n^-}(k,t) \approx -\frac{\ii \eps^n}{\sqrt{2\pi\eps}} \e{-\tfrac{\ii}{\eps}t \hat{H}^-(k)} \int_{-\infty}^t \dd s \, \int_\bbR \dd \eta \, \e{\tfrac{\ii}{\eps}s \hat{H}_1^-(k)}
	\epsft{\kappa_{n+1,0}^{-}} (k - \eta) 
	 \left(\tfrac{\eta+k}{2}\right)^{n+1}
	 \e{-\tfrac{\ii}{\eps}s \hat{H}_1^+(\eta)}\epsft{\phi}(\eta),
\]
where $\hat{H}_1^{\pm}$ ($\hat{H}^\pm$) are the approximate (exact) adiabatic propagators in momentum space. 

By the Avron-Herbst formula, the approximate propagators are given exactly by
\[
	\e{-\frac{\ii}{\eps} s \hat H^\pm_1 (k)} = \e{-\ii \frac{\lambda^2 s^3}{6 \eps}}
	\e{\lambda s \partial_k} \e{-\frac{\ii}{2 \eps} ( (k^2 \pm 2 \delta) s - \lambda k s^2)}.
\]
In particular, we have
\begin{align}
	\label{propagators}
	\e{\tfrac{\ii}{\eps} s \hat{H}_1^-(k)} & = \e{\tfrac{\ii \lambda^2 s^3}{6\eps}} \e{-\lambda s \partial_k} \e{\tfrac{\ii}{2\eps}(k^2-2\delta)s} \e{\tfrac{\ii}{2\eps} \lambda k s^2} \\
	\e{-\tfrac{\ii}{\eps} s \hat{H}_1^+(\eta)} & = \e{-\tfrac{\ii \lambda^2 s^3}{6\eps}} \e{\lambda s \partial_\eta} \e{-\tfrac{\ii}{2\eps}(\eta^2-2\delta)s} \e{\tfrac{\ii}{2\eps} \lambda \eta s^2}. \notag
\end{align}
In order to make use of these expressions we must understand the effects of the shift operators, where $\e{\lambda s \partial_k}f(k)=f(k+\lambda s)$.  Using (\ref{propagators}) in (\ref{psinH1}) we note that, due to the invariance of the integral under $\eta \mapsto \eta - \lambda s$, we may apply the $\eta$ shift to the left with opposite sign.  Hence $\epsft{\kappa_{n+1,0}^-}(k-\eta)$ is unaffected and $(k+\eta)^{n+1} \mapsto (k+\eta-2\lambda s)^{n+1}$.  

Shifting the remaining propagator in $k$ by $-\lambda s$, the remaining multiplicative parts of the propagators are given by
\[
	\exp\big[\tfrac{\ii}{2\eps} \big( [(k-\lambda s)^2 - 2\delta]s + \lambda(k-\lambda s)s^2 -(\eta^2+2\delta)s + \lambda \eta s^2 \big) \big].
\]
Simplifying this expression and inserting it into (\ref{psinH1}) gives
\begin{align}
	\epsft{\psi_n^-}(k,t) &\approx -\frac{\ii \eps^n}{\sqrt{2\pi\eps}} \e{-\tfrac{\ii}{\eps}t \hat{H}_1^-(k)} \int_{-\infty}^t \dd s \, \int_\bbR \dd \eta \, (k+\eta-2\lambda s)^{n+1} \epsft{\kappa_{n+1,0}^-}(k-\eta) \notag \\
	 & \qquad \qquad \times \e{\tfrac{\ii}{2\eps} \big( (k^2-\eta^2-4\delta)s - (k-\eta)\lambda s^2 \big) }\epsft{\phi}(\eta). \label{psinminushat}
\end{align}

\subsection{Fourier transform of the coupling elements} \label{S: FT approx}
In order to make use of (\ref{psinminushat}), we require the Fourier transform of $\kappa_{n,0}^-$.  Using (\ref{epsft}) on (\ref{kappan0}) gives
\begin{align*}
	\epsft{\kappa_{n,0}^{-}}(k) &= \frac{1}{\sqrt{2\pi\eps}} \int \e{-\tfrac{\ii}{\eps} k q} \frac{\ii^{n+1}}{\pi} \rho(q) (n-1)! \Big[ \frac{1}{\big( \tau(q) - \tau_\delta^* \big)^n} - \frac{1}{\big( \tau(q) - \tau_\delta \big)^n} \Big] \dd q \\
	&= \frac{1}{\sqrt{2\pi\eps}} \int \e{-\tfrac{\ii}{\eps} k q(\tau)} \frac{\ii^{n+1}}{2 \pi} (n-1)! \Big[ \frac{1}{\big( \tau - \tau_\delta^* \big)^n} - \frac{1}{\big( \tau - \tau_\delta \big)^n} \Big] \dd \tau,
\end{align*}
where we have used $\dd \tau = 2 \rho(q) \dd q$.

It is now that we need the transition point to be at or near $q=0$. Provided this is so, we can use that $\rho$ has a minimum $\delta$ at 
$q=0$, and expand $q(\tau)=\tfrac{\tau}{2\delta} + \mathcal{O}(\tau^3)$. Note that no second order term is present. As the remainder of the 
integrand is concentrated in a $\sqrt\eps$ neighbourhood around $q=0$, we keep only the first order term, giving
\[
	\epsft{\kappa_{n,0}^{-}}(k) \approx \frac{1}{\sqrt{2\pi\eps}} \int \e{-\tfrac{\ii}{2 \delta \eps} k \tau} \frac{\ii^{n+1}}{2 \pi} (n-1)! 
	\Big[ \frac{1}{\big( \tau - \tau_\delta^* \big)^n} - \frac{1}{\big( \tau - \tau_\delta \big)^n} \Big] \dd \tau.
\]
We now note that
$
	\frac{1}{( \tau - \alpha )^n} = (-1)^{n-1} \frac{1}{(n-1)!}\partial_\tau^{n-1}\frac{1}{\tau-\alpha}
$
and hence
\begin{align*}
	\epsft{\kappa_{n,0}^{-}}(k) &\approx \frac{1}{\sqrt{2\pi\eps}} \frac{\ii^{n+1}}{2 \pi} (-1)^{n-1} \int \e{-\tfrac{\ii}{2 \delta \eps} k \tau}  \partial_\tau^{n-1} \Big[ \frac{1}{\big( \tau - \tau_\delta^* \big)} - \frac{1}{\big( \tau - \tau_\delta \big)} \Big] \dd \tau \\
	&= \frac{1}{\sqrt{2\pi\eps}} \frac{\ii^{n+1}}{2 \pi} (-1)^{n-1} \int \e{-\tfrac{\ii}{2 \delta \eps} k \tau}  \partial_\tau^{n-1} \Big[ \frac{-2 \ii \tau_c}{\big( (\tau-\tau_r) + \tau_c^2 \big)} \Big] \dd \tau.
\end{align*}

Using the identities $\epsft{f}(k)=\tfrac{1}{\sqrt{\eps}} \hat{f}\big(\tfrac{k}{\eps}\big)$, $\widehat{\partial_\tau^n f}(k)=(\ii k)^n \hat{f}(k)$, $\widehat{f(x-a)}(k)=\e{-\ii a k} \hat{f}(k)$ and the standard Fourier transform
\[
	\widehat{\tfrac{a}{x^2+a^2}} (k) = \sqrt{\tfrac{\pi}{2}} \e{-a |k|}
\]
gives
\[
	\epsft{\kappa_{n,0}^{-}}(k) \approx \ii \frac{\sqrt{2}\delta}{\sqrt{\pi \eps}} \frac{1}{(2\delta)^n} \Big( \frac{k}{\eps} \Big)^{n-1} \e{-\tfrac{\tau_c}{2\delta\eps} |k|} \e{-\ii \tfrac{\tau_r}{2\delta\eps} k},
\]
where we have used $\tau_\delta = \tau_r + \ii \tau_c$.  Inserting this formulation into (\ref{psinminushat}) gives
\begin{align}
	  \epsft{\psi_n^-}(k,t) &\approx-\frac{1}{4\pi\eps} \e{-\tfrac{\ii}{\eps}t \hat{H}^-(k)} \int_{-\infty}^t \dd s \, \int_\bbR \dd \eta \, (k+\eta) (1-\tfrac{2\lambda s}{k+\eta})^{n+1} \big( \tfrac{k^2-\eta^2}{4\delta} \big)^n \notag \\
	 & \qquad \qquad \times \e{ -\tfrac{\tau_c}{2\delta\eps} |k-\eta|} \e{ -\tfrac{ \ii \tau_r}{2\delta\eps} (k-\eta)}  \e{\tfrac{\ii}{2\eps} \big( (k^2-\eta^2-4\delta)s - (k-\eta) \lambda s^2 \big) }\epsft{\phi}(\eta). \label{full time evol}
\end{align}

\section{Evaluation of the integral}

\subsection{The choice of $n$}  \label{S:n}
Equation \eqref{full time evol} still depends on the parameter $n$, the order of the superadiabatic representation. For choosing $n$, 
we attempt to use the same argument that was employed in \cite{BGT} in order to obtain universal transition histories. 
The idea then and now 
is that the modulus of the integrand in \eqref{full time evol} depends on $n$, while the phase does not. We will thus try to 
choose $n$ such that stationary phase and maximal modulus occur at the same point, making it possible to perform asymptotic analysis on 
the 
integral. 

We recall 
the assumption that $\tau_r$ is small, and consider the imaginary part of the exponent. Indeed, 
we will set $\tau_{\rm r} =0$ in what follows. This simplifies the analysis and does not seem to greatly 
affect the accuracy of the final result. 
Differentiating the phase of  (\ref{full time evol}) with respect to $s$ and $\eta$ gives
\begin{align}
	(k^2-\eta^2-4\delta) - 2\lambda(k-\eta)s &=0 \label{ds0}\\
	-2\eta s - \lambda s^2 &=0 \label{deta0}.
\end{align}
Note that if $\lambda=0$ then there is only one solution, namely $k^2-\eta^2=4\delta$ and $s=0$; this remains a solution if $\lambda \neq 0$.  

For a simultaneous solution to (\ref{ds0}) and (\ref{deta0}) (i.e. stationary phase for both integrals) we require either $s=0$ and $k^2-
\eta^2-4\delta=0$, or $\lambda s =  - 2 \eta$ and $\eta$ the solution to $-5\eta^2 + 4 k \eta + k^2 - 4\delta=0$.  In the second case, for $k=
\mathcal{O}(1)$, we see that $\eta$ and hence $s$ are also of order 1.  We have already discussed that we expect the significant transitions to 
occur only when $s = \mathcal{O}(\eps^{1/2})$, and we therefore expect this solution to contribute only a negligible amount to the transmitted 
wave packet. So from the stationary phase condition, we obtain $
s = 0$ and $k^2-\eta^2-4\delta=0.$

For the modulus, we assume the case of a Gaussian wave packet of the form \eqref{gauss packet}. Differentiating the logarithm of the 
modulus with respect to $\eta$ and $s$ and equating to zero leads to the equations 
\begin{align}
	(n+1) \frac{2 \lambda}{\eta + k - 2 s} &= 0, \label{ds0 mod} \\ 
	2c(\eta-p_0) - \frac{\tau_{\rm c}}{2 \delta} \eta + n \eps \frac{2 \eta}{k^2-\eta^2} - (n+1) \eps \frac{2 \lambda s}{(k+\eta)
	(k+\eta - 2 \lambda s)}  &= 0. \label{deta0 mod}
\end{align}
Equations \eqref{ds0}--\eqref{deta0 mod} cannot be solved simultaneously, which shows an interesting difference of the present case 
when compared to the non-tilted case treated in \cite{BG} and \cite{BGT}. To make progress, we 
argue that the choice of the 
optimal superadiabatic representation should depend only weakly on the trace $\lambda$ of the potential. Therefore, we allow $\lambda$ to 
vary as well as $n$, $\eta$ and $s$, and obtain the joint solution $s=\lambda=0$, and $n$ and $\eta$ fulfilling 
$n = \frac{\tau_{\rm c}}{\eps k_0}$
with $k_0$ the solution of \eqref{n choice}. We will in future always use this value of $n$, denoted $n_0$.

\subsection{Rescaling} \label{S:Rescaling}
Recall that the wavepacket moves a distance of order 1 in time of order 1, and, for a semiclassical wavepacket, is of width of order $\eps^
{1/2}$.  Hence for times of order $\eps^\gamma$ with $\gamma<1/2$, in position space, the wavefunction is localised well away from the transition 
region.  It follows that there should be little contribution to the integral outside $s \in [-\eps^\gamma, \eps^\gamma]$ for $\gamma<1/2$.  We 
thus restrict the $s$-integral to this region.

We rewrite $(\ref{full time evol})$ as $ \tfrac{1}{4\pi\eps}  \exp(-\tfrac{\ii}{\eps}t \hat{H}^-(k)) \int_\R \dd \eta \, \int_{-\eps^\gamma}^{\eps^\gamma} \dd s \, g(k,\eta,s)$ with 
\begin{align*}
	g(\eta,k,s)=\exp\Big[ &n \log\big(\tfrac{k^2-\eta^2}{4\delta}\big) + \log(k+\eta) + (n+1)\log\big( 1 - \tfrac{2\lambda s}{k+\eta} \big) - \tfrac{\tau_c}{2\delta\epsilon}|k-\eta| \\
	& - i \tfrac{\tau_r}{2\delta\epsilon}(k-\eta) + \tfrac{i}{2\eps}\big[(k^2-\eta^2-4\delta)s - \lambda(k-\eta) s^2 \big] \Big] \epsft{\phi}(\eta).
\end{align*}
We now note that, in order for the phase of the integrand to be stationary in $s$, we expect $\eta \approx \eta^*=\pm\sqrt{k^2-4\delta}$.  For a semiclassical wave packet which has sufficient momentum to move past the avoided crossing, the choice of sign will correspond to the sign of the mean momentum of $\epsft{\phi}$.  For this choice of $\eta^*$ to make sense, it is clear that we require $k^2-4\delta\geq 0$, and so introduce the cutoff function $\chi_{k^2 > 4\delta}$.  The physical meaning of this cutoff is clear when one considers $\eta$ to be the incoming momentum and $k$ the outgoing momentum:  since the potential gap is $2\delta$, by energy conservation we have $k^2/2 = \eta^2/2+2\delta$, and since we require $\eta^2 > 0$ for the wave packet to move past the crossing we have $k^2 > 4\delta$.

We now set $\eta=\tilde\eta \eps^{1/2} + \eta^*$, where $\tilde\eta$ is of order 1 and rescale the $s$ integral by $s=\tilde s \eps^{1/2}$, which causes the domain of the $\tilde s$ integral to be at least of order 1, and tend to the whole of $\R$ as $\eps \to 0$.  Using $\int_\R \dd  \eta \int_{-\eps^\gamma}^{\eps^\gamma}  \dd s \, g(k,\eta,s) =\eps \int_\R \dd \tilde \eta \, \int_{-\eps^{\gamma-1/2}}^{\eps^{\gamma-1/2}}  \dd \tilde s \, g(k,\eps^{1/2}\tilde \eta+\eta^*,\eps^{1/2} \tilde s)$, and removing the tildes from now on, we are interested in
\begin{align}
	&g(\eta\eps^{1/2}+\eta^*,k,\eps^{1/2}s)=\exp\Big[ n \log\big(1-\tfrac{\eta^2\eps + 2\eta\eta^*\eps^{1/2}}{4\delta}\big) + \log(k+\eta^*+\eta\eps^{1/2}) \notag \\
	& \qquad + (n+1)\log\big( 1 - \tfrac{2\lambda s \eps^{1/2}}{k+\eta^*+\eta\eps^{1/2}} \big) 
	-\tfrac{\tau_c}{2\delta\epsilon}|k-\eta^*-\eta\eps^{1/2}|  
	-\ii \tfrac{\tau_r}{2\delta\epsilon}(k-\eta^*-\eta\eps^{1/2}) \notag\\
	& \qquad + \tfrac{\ii}{2\eps}\big[( -\eta^2\eps - 2\eta \eta^* \eps^{1/2})s\eps^{1/2} - \lambda(k-\eta^*-\eta\eps^{1/2}) s^2 \eps \big] \Big] \epsft{\phi}(\eps^{1/2} \eta + \eta^*). \label{rescaledg}
\end{align}
We now discuss the evaluation of these two integrals.

\subsection{The $s$ integral}

Since the wave function $\epsft{\phi}$ is independent of $s$, we now aim to perform the $s$-integration explicitly.  We now consider the regime where $\eps$ is small and $k$ is of order 1. This  is necessary as we wish to expand the logarithm term in powers of $s$, and require that $ \tfrac{2\lambda s \eps^{1/2}}{k+\eta^* + \eta\eps^{1/2}} \ll 1$.  This holds since, from the limits of integration, we see that at worst $s\eps^{1/2} \sim \eps^{\gamma}$ with $\gamma>0$ and $\eta^*\sim k \sim \eta \sim 1$. Expanding to second order gives
\begin{align}
	\log\big( 1 - \tfrac{2\lambda s \eps^{1/2}}{k+\eta^* + \eta\eps^{1/2}} \big) &\approx -\tfrac{2\lambda s \eps^{1/2}}{k+\eta^* + \eta\eps^{1/2}}  - \tfrac{2\lambda^2 s^2 \eps}{(k+\eta^* + \eta\eps^{1/2})^2}. \label{sLogExpansion}
\end{align}
%
In the small-$\eps$ limit, $\eps^{\gamma-1/2} \to \infty$, which, combined with the above expansion reduces the $s$-integral to a Gaussian integral of the form
\[
	\int_\R \exp(\alpha s^2 + \beta s) \dd s = \sqrt{-\frac{\pi}{\alpha}} \exp\Big( {-\frac{\beta^2}{4 \alpha}} \Big), \; \mbox{ for } \Re(\alpha)<0.
\]
In this case, we have 
\[
	\alpha = -\tfrac{2(n+1)\lambda^2 \eps}{(k+\eta^*+\eps^{1/2}\eta)^2} - \tfrac{\ii \lambda}{2}(k-\eta^*) +\tfrac{\ii \eps^{1/2} \lambda}{2} \eta
\] 
(where $ \Re(\alpha)<0$) and 
\[
	\beta=-\tfrac{2\lambda(n+1)\eps^{1/2}}{k+\eta^*+\eps^{1/2}\eta} - \ii \eta^*\eta - \tfrac{\ii \eps^{1/2}}{2}\eta^2.
\]  
It therefore remains to calculate the integral over $\eta$:
\begin{align*}
& \epsft{\psi_n^-}(k,t)  = \frac{\chi_{k^2 > 4\delta}}{4 \sqrt \pi \eps} \e{-\frac{\ii}{\eps} t \hat{H}^-(k)} 
\int_\R \dd \eta \,  \epsft{\phi}(\eps^{1/2} \eta + \eta^*) \Big( \tfrac{2(n+1)\lambda^2 \eps}{(k+\eta^*+\eps^{1/2}\eta)^2} + \tfrac{\ii \lambda}{2}(k-\eta^*) -\tfrac{\ii \eps^{1/2} \lambda}{2} \eta \Big)^{-1/2}\\
&\times \exp\Big[ n \log\big(1-\tfrac{\eta^2\eps + 2\eta\eta^*\eps^{1/2}}{4\delta}\big) + \log(k+\eta^*+\eta\eps^{1/2}) 
-\tfrac{\tau_c}{2\delta\epsilon}|k-\eta^*-\eta\eps^{1/2}|  
-\ii \tfrac{\tau_r}{2\delta\epsilon}(k-\eta^*-\eta\eps^{1/2})  \Big] \\
	& \times \exp\Big[-\Big(-\tfrac{2\lambda(n+1)\eps^{1/2}}{k+\eta^*+\eps^{1/2}\eta} - \ii \eta^*\eta - \tfrac{\ii \eps^{1/2}}{2}\eta^2\Big)^2  
\Big(  -\tfrac{2(n+1)\lambda^2 \eps}{(k+\eta^*+\eps^{1/2}\eta)^2} - \tfrac{\ii \lambda}{2}(k-\eta^*) +\tfrac{\ii \eps^{1/2} \lambda}{2} \eta \Big)^{-1} \Big].
\end{align*}

For a general $\epsft{\phi}$, we can say little else, and the integral must be computed numerically.  However, in the important case where $\epsft{\phi}$ is a Gaussian, we can derive a closed-form approximation, which is in excellent agreement with the full dynamics.  The main idea is to approximate the integrand in (\ref{full time evol}) in such a way as to produce a Gaussian integral.  The first hinderance to this comes from the $\log$ terms, which we now consider.

\subsection{Expansion of $\log$ terms}

Along with the expansion in  (\ref{sLogExpansion}), we have
\[
	\log(k+\eta^*+\eta \eps^{1/2}) = \log(k+\eta^*) + \log\big( 1 + \tfrac{\eta\eps^{1/2}}{k+\eta^*} \big) \approx  \log(k+\eta^*)  + \tfrac{\eta\eps^{1/2}}{k+\eta^*} - \tfrac{\eta^2 \eps}{2(k+\eta^*)^2},
\]
\[
	\log\big( 1- \tfrac{\eta^2\eps + 2 \eta \eta^* \eps^{1/2}}{4\delta} \big) \approx -\tfrac{\eta^2\eps + 2 \eta \eta^* \eps^{1/2}}{4\delta} - \tfrac{1}{32\delta^2}( \eta^4\eps^2 + 4 \eta^3\eta^*\eps^{3/2} + 4\eta^2\eta^{*2}\eps),
\]
where we have once again used that $k, \eta^*, \eta  \sim 1$.  

In order to produce a Gaussian integral, it is necessary to make a number of justifiable approximations.  Expanding $(k+\eta^*+\eps^{1/2}\eta)^{-p}$, $p=1,2$ in (\ref{sLogExpansion}) around $\eta\eps^{1/2}=0$ and neglecting terms of order larger than $\eps$ in all three logarithm expansions reduces them to:
\begin{align}
	\log\big( 1 - \tfrac{2\lambda s \eps^{1/2}}{k+\eta^* + \eta\eps^{1/2}} \big) &\approx -\tfrac{2\lambda s \eps^{1/2}}{k+\eta^*} + \tfrac{2\lambda s \eps \eta }{(k+\eta^*)^2} - \tfrac{2\lambda^2 s^2 \eps}{(k+\eta^*)^2}, \notag \\
	\log(k+\eta^*+\eta \eps^{1/2}) &\approx  \log(k+\eta^*)  + \tfrac{\eta\eps^{1/2}}{k+\eta^*} - \tfrac{\eta^2 \eps}{2(k+\eta^*)^2}, \label{finalLogExpansions}\\
	\log\big( 1- \tfrac{\eta^2\eps + 2 \eta \eta^* \eps^{1/2}}{4\delta} \big) &\approx -\tfrac{\eta^2\eps + 2 \eta \eta^* \eps^{1/2}}{4\delta} - \tfrac{\eta^2\eta^{*2}\eps}{8\delta^2}. \notag
\end{align}
Note that all three expansions now contain terms of at most order two in $s$ and $\eta$ and thus are of the form required for a Gaussian integral.  

\subsection{Explicit closed form}

One final simplification is necessary to obtain a Gaussian integral: \eqref{rescaledg} still contains the third order terms,
namely $\tfrac{\ii \eps^{1/2}}{2} \eta^2 s$ and $\tfrac{\ii \lambda}{2 \eps^{1/2}}\eta s^2$. But here we recall that the staionary
phase argument required $s=0$, and in the scaled variables also $\eta=0$. 
This allows us to remove the above terms: not only are these terms already the highest order in $\eps$, 
but since we expect the main contribution to the integral to come from the region around $(s,\eta)=(0,0)$, the effects of these terms is 
negligible.  

Inserting the expansions (\ref{finalLogExpansions}) into (\ref{rescaledg}), ignoring the third order terms in $s$ and $\eta$, 
and setting $\epsft{\phi}(\eta)$ to be the Gaussian $\epsft{\phi}(\eta)=\exp(-\tfrac{c}{\eps}(\eta-p_0)^2)$ gives, 
for $\eta$ sufficiently small,
\[
g(k,\eta\eps^{1/2}+\eta^*,\eps^{1/2}s) = \exp( \alpha_{2,0} \eta^2 + \alpha_{1,0} \eta + \alpha_{1,1} \eta s + \alpha_{0,1} s + \alpha_{0,2} s^2 ),
\]
with the $\alpha_{i,j}$ given in \eqref{alphas}. 
Note that the $\sgn(k)$ in $\alpha_{1,0}$ is necessary if we wish to deal with negative momenta:  for $k>0$, we have $k-\eta^*>0$ and hence, for small $\eps$, $k-
\eta^*-\eps^{1/2}\eta>0$.  Therefore $|k-\eta^*-\eps^{1/2}\eta| = |k-\eta^*|-\eps^{1/2}\eta$.  However, for $k<0$ we have $k-\eta^*-\eps^
{1/2}\eta<0$ and $|k-\eta^*-\eps^{1/2}\eta| = |k-\eta^*|+\eps^{1/2}\eta$.

Gaussian integration now gives 
\be
	 \int_\R \int_\R \dd \eta  \, \dd s \; g(k,\eta\eps^{1/2}+\eta^*,\eps^{1/2}s) = \frac{2\pi}{\sqrt{4 \alpha_{2,0}\alpha_{0,2} - \alpha_{1,1}^2}} \exp\Big[ \frac{ \alpha_{2,0}\alpha_{0,1}^2 + \alpha_{0,2}\alpha_{1,0}^2 - \alpha_{1,0}\alpha_{0,1}\alpha_{1,1} } {\alpha_{1,1}^2 - 4\alpha_{2,0}\alpha_{0,2}}\Big], 
\ee
which holds for $\Re \big(\tfrac{\alpha_{1,1}^2}{\alpha_{2,0}}-4\alpha_{0,2})>0$.  

We now check that the above constraint $\Re \big(\tfrac{\alpha_{1,1}^2}{\alpha_{2,0}}-4\alpha_{0,2})>0$ is satisfied for a suitable parameter regime.  For ease of analysis, we note that $n_0$ is approximately given by $\tau_c/(\eps \sqrt{p_0^2+4\delta}) = \mathcal{O}(\eps^{-1})$.  Taking $\eps$ to be small, to leading order we find
\begin{align*}
 	\alpha_{2,0} &=-\tfrac{n_0\eps}{4\delta} - \tfrac{n_0 \eta^{*2}\eps}{8\delta^2} -c \approx -\tfrac{\tau_c}{4\delta \sqrt{p_0^2+4\delta}} - \tfrac{\tau_c \eta^{*2}}{8 \delta^2 \sqrt{p_0^2+4\delta}}  - c\\
	\alpha_{1,1} &= -\ii \eta^* + \tfrac{2 n_0 \eps \lambda}{(k+\eta^*)^2} \approx -\ii \eta^* + \tfrac{2 \tau_c \lambda}{(k+\eta^*)^2\sqrt{p_0^2+4\delta}}\\
	\alpha_{0,2}&= -\ii \tfrac{2\delta\lambda}{(k+\eta^*)}  - \tfrac{2(n_0+1)\lambda^2 \eps}{(k+\eta^*)^2} \approx -\ii \tfrac{2\delta \lambda}{(k+\eta^*)}  - \tfrac{2\tau_c\lambda^2}{(k+\eta^*)^2\sqrt{p_0^2+4\delta}}.
\end{align*}
Note that the real part of $-4\alpha_{0,2}$ is non-negative, so we need only check the sign of $\Re(\alpha_{1,1}^2/\alpha_{2,0})$.  Using $\alpha_{1,1}^2 = -\eta^{*2} + \tfrac{4 \tau_c^2 \lambda^2}{(k+\eta^*)^4(p_0^2+4\delta)} -  \ii \tfrac{4\tau_c\lambda \eta^*}{(k+\eta^*)^2\sqrt{p_0^2+4\delta}}$ gives
\[
	\Re(\tfrac{\alpha_{1,1}^2}{\alpha_{2,0}} - 4 \alpha_{0,2}) \geq \tfrac{8\delta^2}{(k+\eta^*)^4\sqrt{p_0^2+4\delta}} \Big[ \tfrac{ \eta^{*2}(k+\eta^*)^4(p_0^2+4\delta) - 4\tau_c^2\lambda^2}{8\delta^2 c \sqrt{p_0^2+4\delta} + 2\delta \tau_c + \tau_c p_0^2} \Big].
\]
Since $\tau_c>0$, this is clearly positive when $p_0$ is sufficiently large.  Hence the regime of interest is $\eps$ small and $p_0$ large.  We then have
\begin{align}
	\epsft{\psi_n^-}(k,t)  \approx  & \e{-\frac{\ii}{\eps} t \hat{H}^-}  \frac{1}{2\sqrt{4 \alpha_{2,0}\alpha_{0,2} - \alpha_{1,1}^2}} \exp\Big[ \frac{ \alpha_{2,0}\alpha_{0,1}^2 + \alpha_{0,2}\alpha_{1,0}^2 - \alpha_{1,0}\alpha_{0,1}\alpha_{1,1} } {\alpha_{1,1}^2 - 4\alpha_{2,0}\alpha_{0,2}}\Big] \notag\\
	& \times (\eta^*+k) \e{-\tfrac{c}{\eps}(\eta^*-p_0)^2} \e{-\tfrac{\tau_c}{2\delta\epsilon}|k-\eta^*|} \e{-\ii \tfrac{\tau_r}{2\delta\epsilon}(k-\eta^*)} \chi_{k^2 > 4\delta}, \label{explicitForm}
\end{align}
with the $\alpha_{i,j}$ as given in (\ref{alphas}).

We note that setting $\lambda=0$ gives $\alpha_{0,1}=\alpha_{0,2}=0$ and $\alpha_{1,1}=\ii \eta^*$, and returns the $n$-independent form (see \cite{BGT})
\[
	\epsft{\psi_n^-}(k,0)  \approx  \tfrac{(\eta^*+k)}{2|\eta^*|} \e{-\tfrac{c}{\eps}(\eta^*-p_0)^2}\e{-\tfrac{\tau_c}{2\delta\epsilon}|k-\eta^*|} \e{-\ii \tfrac{\tau_r}{2\delta\epsilon}(k-\eta^*)} \chi_{k^2 > 4\delta}.
\]

\subsection{Asymptotics of $\alpha_{i,j}$} \label{S:AlphaApprox}

We now show that, under suitable assumptions, the $\alpha_{i,j}$ given in (\ref{alphas}) may be somewhat simplified.  Note that $n_0=\mathcal{O}(\eps^{-1})$ and so $\alpha_{2,0}=\mathcal{O}(1)$ and $\alpha_{1,0}=\mathcal{O}(\eps^{-1/2})$.  However, the two terms which come from the $n$-independent prefactor $k+\eta$ in (\ref{full time evol})  are of lower order than the remaining terms in the respective $\alpha_{i,j}$, and hence, for small $\eps$ may safely be neglected.  From the point of view of exponential asymptotics this is completely natural; one would normally fix the slowly varying terms (i.e. independent of $\eps$, and in this case of $n$) at the stationary value of the integrand.  For clarity, we now have
\[
\alpha_{2,0}=  -\tfrac{n_0\eps}{4\delta} - \tfrac{n_0 \eta^{*2}\eps}{8\delta^2}  -c, \qquad 
	\alpha_{1,0}= - \tfrac{n_0\eta^* \eps^{1/2} }{2\delta} - \tfrac{2c (\eta^* -p_0)}{\eps^{1/2}} +  \tfrac{\sgn(k)\tau_c}{2\delta\eps^{1/2}} +\ii  \tfrac{\tau_r}{2\delta\eps^{1/2}}.
\]

One additional simplification is possible when $p_0$ is large and the wavefunction is quickly decaying (i.e. $c$ is also large).  In this case the modulus of the integrand is negligible unless $\eta^*$ is close to $p_0$.  For such a range, using $n_0\approx \tau_c/(\eps \sqrt{p_0^2+4\delta}) $, shows that the first three terms in $\alpha_{1,0}$ above are all negligible.  Further, if the potential is symmetric, $\tau_r=0$ and we may use the approximation $\alpha_{1,0}=0$.
In addition, in this limit, and with the assumption that $\lambda$ is not too large, we see that the second terms in each of $\alpha_{1,1}$ and $\alpha_{0,2}$ in (\ref{alphas}) are negligible.  To conclude, for $\eps$ small, $p_0$ and $c$ large and $\lambda$ not too large, we have
\begin{align*}
	&\alpha_{2,0}\approx  -\frac{\tau_c (2 \delta + {\eta^\ast}^2)}{8\delta^2\eta^*} -c, \quad 
	\alpha_{1,0}\approx   \frac{ \ii \tau_r}{2\delta\sqrt{\eps}}, \quad 
	\alpha_{1,1}\approx -\ii \eta^*, \notag\\
	&\alpha_{0,1}\approx -\frac{2 \lambda}{\sqrt{\eps}(k+\eta^*)\eta^*}, \quad
	\alpha_{0,2}\approx -\frac{ \ii 2\delta \lambda}{(k+\eta^*)}. \notag
\end{align*}

\subsection{Additional Phase shift} \label{S:PhaseShift}
While testing the formula (\ref{explicitForm}) against ab-initio numerics, we found a discrepancy by a phase shift which, in the region 
where the wavefunction has significant magnitude, is constant in $k$. We believe that this effect comes from one of the approximations detailed 
above, but have currently been unable to determine its exact cause.  For many applications this phase shift is unimportant. Since it is constant 
in $k$, all expected values of observables are correctly reproduced by \eqref{explicitForm} in the case of a single Gaussian wave packet. 
Where the phase shift begins to matter is for interference phenomena, and when considering a superposition of Gaussians (see below)
such that their centres are at significantly different locations in $k$; then the phase shift will not be constant in $k$ any more, and 
we will get wrong predictions for position expected values. 

It is therefore desirable to have a method of removing this effect of the approximations.  We now describe a heuristic method which has 
proven to be effective for a wide range of potentials and initial Gaussian wave packets.
Consider (\ref{explicitForm}) for $\lambda \neq 0$ and the wavefunction normalized by a prefactor $\sqrt{c/(\pi \eps)}$.  Note that if $
\lambda=0$ the following argument is invalid.  However, setting $\lambda=0$ in (\ref{explicitForm}) we see that the phase depends only on $\tau_r
$, which agrees with that of \cite{BG} and the corresponding numerics.

We are going to consider the phase of the transmitted wave function in the limit $c \to \infty$, i.e. the incoming wave packet approximating a 
$\delta$-function at $\eta=p_0$.  Since the numerical phase shift is independent of $k$, we need to choose a value of $k$ at which to evaluate 
this phase. In the classical picture, from energy conservation we see that the the transmitted wave packet should be approximately a $\delta$-
function at $k=\sqrt{p_0^2+4\delta}$, and hence we consider this value of $k$, where the sign of the square root is chosen to match that of 
$p_0$.

We are therefore interested in
\begin{align*}
	\epsft{\psi_n^-}(\sqrt{p_0^2+4\delta},0)  \approx  & \frac{\sqrt{c}}{\sqrt{\pi \eps}} \frac{1}{2\sqrt{4 \alpha_{2,0}\alpha_{0,2} - \alpha_{1,1}^2}} \exp\Big[ \frac{ \alpha_{2,0}\alpha_{0,1}^2 + \alpha_{0,2}\alpha_{1,0}^2 - \alpha_{1,0}\alpha_{0,1}\alpha_{1,1} } {\alpha_{1,1}^2 - 4\alpha_{2,0}\alpha_{0,2}}\Big]\\
	& \times a_0 \e{ \tfrac{- \tau_c}{2\delta\eps} |b_0|} \e{ \tfrac{- \ii \tau_r}{2\delta\eps} b_0} \chi_{k^2 > 4\delta},
\end{align*}
where $a_0=\sqrt{p_0^2+4\delta} + p_0$ and $b_0=\sqrt{p_0^2+4\delta} - p_0$. We now investigate the phase of this wave packet when $c \to \infty$ and note that there are contributions from both the square root and the exponent.

We write $\alpha_{2,0}=\beta_{2,0} -c$, and, since $\eta^*=p_0$, this is the only term that depends on $c$.  Consider first the prefactor:
\begin{align}
	\frac{\sqrt{c}}{\sqrt{4\alpha_{2,0}\alpha_{0,2}-\alpha_{1,1}^2}} &= \frac{\sqrt{c}}{\sqrt{4(\beta_{2,0}-c)\alpha_{0,2}-\alpha_{1,1}^2}}  = \frac{1}{\sqrt{-4\alpha_{0,2}+ \tfrac{1}{c}(4\beta_{2,0}\alpha_{0,2}-\alpha_{1,1}^2)}} \notag \\
	&\mathop{\to}^{c\to \infty} \frac{1}{\sqrt{-4 \alpha_{0,2}}}. \label{sqrtLimit}
\end{align}
For the exponent, we have
\begin{align}
	& \frac{\alpha_{2,0}\alpha_{0,1}^2+\alpha_{0,2}\alpha_{1,0}^2-\alpha_{1,0}\alpha_{0,1}\alpha_{1,1}} {\alpha_{1,1}^2-4\alpha_{2,0}\alpha_{0,2}} 
	= \frac{(\beta_{2,0}-c)\alpha_{0,1}^2+\alpha_{0,2}\alpha_{1,0}^2-\alpha_{1,0}\alpha_{0,1}\alpha_{1,1}}{\alpha_{1,1}^2-4(\beta_{2,0}-c)\alpha_{0,2}} \notag \\
	& \qquad = \frac{ -\alpha_{0,1}^2 + \tfrac{1}{c}( \beta_{2,0}\alpha_{0,1}^2 + \alpha_{0,2}\alpha_{1,0}^2 - \alpha_{1,0}\alpha_{0,1}\alpha_{1,1}) }{ 4 \alpha_{0,2} + \tfrac{1}{c}(\alpha_{1,1}^2 - 4\beta_{2,0}\alpha_{0,2})} 
	\mathop{\to}^{c\to \infty}  -\frac{\alpha_{0,1}^2}{4\alpha_{0,2}} \label{expLimit}.
\end{align}

It remains to determine the phases of (\ref{sqrtLimit}) and (\ref{expLimit}).  We write $\alpha_{1,0}=\beta_{1,0}$, $\alpha_{1,1}=\beta_{1,1}+\ii \gamma_{1,1}$, $\alpha_{0,1}=\beta_{0,1}$ and $\alpha_{0,2}=\beta_{0,2}+\ii \gamma_{0,2}$, with $\beta_{i,j}, \gamma_{i,j} \in \R$.  For (\ref{sqrtLimit}), we note that
$
	-\alpha_{0,2}=-(\beta_{0,2}+\ii \gamma_{0,2} )=: r\e{\ii \theta},
$
where $\theta = \arctan \Big(\tfrac{\gamma_{0,2}}{\beta_{0,2}}\Big)$, giving the phase of (\ref{sqrtLimit}) as
$
	-\frac{1}{2}\arctan \Big(\frac{\gamma_{0,2}}{\beta_{0,2}}\Big).
$
Using that $\gamma_{0,2}=-\lambda b_0 /2$, $\beta_{0,2}=-2(n_0+1)\lambda^2\eps/a_0^2$ and$a_0b_0=4\delta$ , this is
$
	-\frac{1}{2}\arctan \Big(\frac{a_0 \delta}{(n_0+1)\eps\lambda}\Big).
$

For the  (\ref{expLimit}) we have
$
	-\frac{\beta_{0,1}^2}{4(\beta_{0,2}+\ii \gamma_{0,2})} = -\frac{\beta_{0,1}^2(\beta_{0,2}-\ii \gamma_{0,2})}{4(\beta_{0,2}^2+ \gamma_{0,2}^2)}.
$
Hence, using that $\beta_{0,1}=2(n_0+1)\eps^{1/2}\lambda/a_0$, this phase is given by 
$
	-\frac{(n_0+1)^2\eps\lambda a_0 \delta}{2(n+1)^2\lambda^2 \eps^2 + 2\delta^2 a_0^2},
$
and the total phase by
\[
	-\frac{(n_0+1)^2\eps\lambda a_0 \delta}{2(n_0+1)^2\lambda^2 \eps^2 + 2\delta^2 a_0^2}
	+ \frac{1}{2}\arctan \Big(-\frac{a_0 \delta}{(n_0+1)\eps\lambda}\Big) - \frac{\tau_r}{2\delta\eps}b_0.
\]

One further adjustment seems to be necessary.  One would expect that the phase is continuous in $\lambda$, and we know that for $\lambda=0$, the phase is $- \frac{\tau_r}{2\delta\eps}b_0$.  However, the limit of the $\lambda$-dependent terms in  above expression is $-1/2 \arctan(\sgn(\lambda)\sgn(a_0)\infty)=-\sgn(\lambda)\sgn(p_0)\pi/4$ and hence we take the phase shift to be
\be
	\varphi(p_0)=-\frac{(n_0+1)^2\eps\lambda a_0 \delta}{2(n_0+1)^2\lambda^2 \eps^2 + 2\delta^2 a_0^2}
	- \frac{1}{2}\arctan \Big(\frac{a_0 \delta}{(n_0+1)\eps\lambda}\Big) + \sgn(\lambda p_0) \frac{\pi}{4}, \label{phaseShift}
\ee
which seems to give very good numerical results for a wide range of all parameters. 

To summarize, we now have an explicit closed form for the transmitted wave packet given an initial Gaussian of the form
\eqref{gauss packet}: 
\begin{align}
	\epsft{\psi_n^-}(k,t)  \approx  & \e{-\frac{\ii}{\eps} t \hat{H}^-}  \frac{1}{2\sqrt{4 \alpha_{2,0}\alpha_{0,2} - \alpha_{1,1}^2}} \exp\Big[ \frac{ \alpha_{2,0}\alpha_{0,1}^2 + \alpha_{0,2}\alpha_{1,0}^2 - \alpha_{1,0}\alpha_{0,1}\alpha_{1,1} } {\alpha_{1,1}^2 - 4\alpha_{2,0}\alpha_{0,2}}\Big] \notag\\
	& \times (\eta^*+k)  \e{-\tfrac{\tau_c}{2\delta\epsilon}|k-\eta^*|} \e{-\ii \tfrac{\tau_r}{2\delta\epsilon}(k-\eta^*)} \e{-\ii \varphi(p_0)} \epsft{\phi}(\eta^*) \chi_{k^2 > 4\delta},
\end{align}
with $\varphi(p_0)$ as given in (\ref{phaseShift}) and the $\alpha_{i,j}$ as in (\ref{alphas}), or, alternatively, with the simplifications 
from Section \ref{S:AlphaApprox}. 
$n_0$ is given as indicated in \eqref{n choice}. We thus have finished the justification of the algorithm given in Section
\ref{S:algorithm}.

\subsection{Phase shift for large momentum}

For large momentum, we use the approximations $a_0=2p_0$ and $n_0=\tau_c/(\eps p_0)$ and get the phase shift $\varphi(p_0)$ as
$
	\varphi(p_0) \approx -\frac{1}{\eps}\frac{\tau_c^2 \lambda \delta p_0}{\tau_c^2\lambda^2+4\delta^2p_0^4} - \frac{1}{2}\arctan\Big( \frac{4 p_0^2\delta}{\tau_c \lambda} \Big) + \sgn(\lambda p_0) \frac{\pi}{4}
$
Note that if $p_0\to\infty$ then $\varphi(p_0) \to 0$.  More concretely, we are interested in the rate at which it goes to zero when we write $p_0$ in terms of $\eps$. Letting $p_0=\eps^{-\alpha}$ gives
\[
	\varphi \approx - \frac{1}{\eps}\frac{\tau_c^2 \lambda \delta}{\tau_c^2\lambda^2\eps^{1+\alpha}+4\delta^2 \eps^{1-3\alpha}} - \frac{1}{2}\arctan\Big( \frac{4\eps^{-2\alpha} \delta}{\tau_c \lambda} \Big) + \sgn(\lambda p_0) \frac{\pi}{4}.
\]
Hence if $\alpha>1/3$ we see that $\varphi(\eps^{-\alpha})\to 0$ as $\eps \to 0$.  We note that this value of $1/3$ is the same value as that for which we have rigorous bounds on the errors \cite{BG2}.

From this analysis, it appears that the phase shift is a consequence of taking momenta that are too small (or equivalently, $\eps$ that are too large).

\section{Non-Gaussian incoming wavefunctions}

\subsection{Extension to Hagedorn wavefunctions}
We note that a general Hagedorn wavefunction \cite{Ha94} is a Hermite polynomial multiplied by a Gaussian.  By linearity of the integral, it is sufficient to consider the case $\epsft{\phi}(\eta)=\eta^p \exp\big({-c}(\eta-p_0)^2/\eps\big)$, $p\in \N$.  We perform the same rescaling as in Section \ref{S:Rescaling} and note that the monomial prefactor becomes
$
	(\eta\eps^{1/2}+\eta^*)^p=\sum_{j=0}^p \binom{p}{j} (\eta\eps^{1/2})^j \eta^{*(p-j)}.
$
Using the same arguments as above, we obtain for each $j$ the integral
\begin{align*}
	& \int _\R  \int_\R \dd \eta \dd s \; (\eps^{1/2} \eta)^j \exp( \alpha_{2,0} \eta^2 + \alpha_{1,0} \eta + \alpha_{1,1} \eta s + \alpha_{0,1} s + \alpha_{0,2} s^2 ) 
\end{align*}

We now note that $ \partial_{\alpha_{1,0}}^j \exp( \alpha_{1,0} \eta ) = \eta^j  \exp(\alpha_{1,0} \eta )$ and since differentiation with respect to $\alpha_{1,0}$ commutes with the integral, we have
\begin{align*}
	& \int _\R  \int_\R \dd \eta  \, \dd s \;  (\eps^{1/2} \eta)^j \exp( \alpha_{2,0} \eta^2 + \alpha_{1,0} \eta + \alpha_{1,1} \eta s + \alpha_{0,1} s + \alpha_{0,2} s^2 ) \\
	& \qquad = \eps^j \partial_{\alpha_{1,0}}^j \frac{2\pi}{\sqrt{4 \alpha_{2,0}\alpha_{0,2} - \alpha_{1,1}^2}} \exp\Big[ \frac{ \alpha_{2,0}\alpha_{0,1}^2 + \alpha_{0,2}\alpha_{1,0}^2 - \alpha_{1,0}\alpha_{0,1}\alpha_{1,1} } {\alpha_{1,1}^2 - 4\alpha_{2,0}\alpha_{0,2}}\Big] \\
	& \qquad = \eps^j  \frac{2\pi}{\sqrt{4 \alpha_{2,0}\alpha_{0,2} - \alpha_{1,1}^2}} 
	\exp \Big[ \frac{ \alpha_{2,0}\alpha_{0,1}^2} {\alpha_{1,1}^2 - 4\alpha_{2,0}\alpha_{0,2}}\Big] \partial_{\alpha_{1,0}}^j
	\exp\Big[ \frac{\alpha_{0,2}\alpha_{1,0}^2 - \alpha_{1,0}\alpha_{0,1}\alpha_{1,1} } {\alpha_{1,1}^2 - 4\alpha_{2,0}\alpha_{0,2}}\Big].
\end{align*}

In more generality, we wish to compute $\partial_\alpha^j f$ where $f=\exp (-a \alpha^2 + b\alpha)=\exp( -\tfrac{2 a}{2}(\alpha-\tfrac{b}{2a})^2 + \tfrac{b^2}{4a})$.  It is clear that this will be $f$ multiplied by a scaled and shifted Hermite polynomial.  In fact, we have 
$
	\partial_\alpha^j f = (-\sqrt{2 a})^j H_j\big(\sqrt{2a} (\alpha - \tfrac{b}{2a})\big) f,
$
where $H_j$ is the probabilist's Hermite polynomial of order $j$ (namely chosen such that the coefficient of the leading order is 1).

In our case, we have $a=-\tfrac{\alpha_{2,0}}{\alpha_{1,1}^2-4\alpha_{2,0}\alpha_{0,2}}$ and $b=-\tfrac{\alpha_{0,1}\alpha_{1,1}}{\alpha_{1,1}^2-4\alpha_{2,0}\alpha_{0,2}}$, giving $\tfrac{b}{2a}=\tfrac{\alpha_{0,1}\alpha_{1,1}}{2\alpha_{0,2}}$.  Hence (\ref{full time evol}) with $\epsft{\phi}(\eta)=\eta^p \exp\big({-c}(\eta-p_0)^2/\eps\big)$ is given by
\begin{align*}
	& \epsft{\psi_n^-}(k,t)  \approx  \e{-\frac{\ii}{\eps} t \hat{H}^-(k)}  \frac{\chi_{k^2 > 4\delta}}{2\sqrt{4 \alpha_{2,0}\alpha_{0,2} - \alpha_{1,1}^2}} \exp\Big[ \frac{ \alpha_{2,0}\alpha_{0,1}^2 + \alpha_{0,2}\alpha_{1,0}^2 - \alpha_{1,0}\alpha_{0,1}\alpha_{1,1} } {\alpha_{1,1}^2 - 4\alpha_{2,0}\alpha_{0,2}}\Big]\\
	& \qquad \times (\eta^*+k) \e{-\tfrac{c}{\eps}(\eta^*-p_0)^2} \e{-\tfrac{\tau_c}{2\delta\epsilon}|k-\eta^*|} \e{-\ii \tfrac{\tau_r}{2\delta\epsilon}(k-\eta^*)}\\
	& \qquad \times \sum_{j=0}^p \binom{p}{j} \eps^j \eta^{*(p-j)} \Big(\frac{2\alpha_{0,2}}{\alpha_{1,1}^2-4\alpha_{2,0}\alpha_{0,2}}\Big)^{j/2} 
	 \times H_j\Big[\Big(\frac{2\alpha_{0,2}}{\alpha_{1,1}^2-4\alpha_{2,0}\alpha_{0,2}}\Big)^{1/2}  \Big(\alpha_{1,0}-\frac{\alpha_{0,1}\alpha_{1,1}}{2\alpha_{0,2}}\Big)\Big],
\end{align*}
with the $\alpha_{i,j}$ as given in (\ref{alphas}).

Using the identity
$
	H_p(x+y) = x^p \sum_{j=0}^p \binom{p}{j} x^{-j}H_j(y),
$
with
$
	x = \tfrac{\eta^*}{\eps} \big(\tfrac{2\alpha_{0,2}}{\alpha_{1,1}^2-4\alpha_{2,0}\alpha_{0,2}}\big)^{-1/2}$, and $y=\big(\tfrac{2\alpha_{0,2}}{\alpha_{1,1}^2-4\alpha_{2,0}\alpha_{0,2}}\big)^{1/2}  \big(\alpha_{1,0}-\tfrac{\alpha_{0,1}\alpha_{1,1}}{2\alpha_{0,2}} \big)
$
gives
\begin{align*}
	\epsft{\psi_n^-}(k,t)  \approx  & \e{-\frac{\ii}{\eps} t \hat{H}^-(k)}  \frac{\chi_{k^2 > 4\delta}}{2\sqrt{4 \alpha_{2,0}\alpha_{0,2} - \alpha_{1,1}^2}} \exp\Big[ \frac{ \alpha_{2,0}\alpha_{0,1}^2 + \alpha_{0,2}\alpha_{1,0}^2 - \alpha_{1,0}\alpha_{0,1}\alpha_{1,1} } {\alpha_{1,1}^2 - 4\alpha_{2,0}\alpha_{0,2}}\Big]\\
	& \times (\eta^*+k) \e{-\tfrac{c}{\eps}(\eta^*-p_0)^2}\e{-\tfrac{\tau_c}{2\delta\epsilon}|k-\eta^*|} \e{-\ii \tfrac{\tau_r}{2\delta\epsilon}(k-\eta^*)} \eps^p \Big(\frac{2\alpha_{0,2}}{\alpha_{1,1}^2-4\alpha_{2,0}\alpha_{0,2}}\Big)^{p/2}   \\
	&\times H_j\Big[ \frac{\eta^*}{\eps} \Big(\frac{2\alpha_{0,2}}{\alpha_{1,1}^2-4\alpha_{2,0}\alpha_{0,2}}\Big)^{-1/2} +   \Big(\frac{2\alpha_{0,2}}{\alpha_{1,1}^2-4\alpha_{2,0}\alpha_{0,2}}\Big)^{1/2}  \Big(\alpha_{1,0}-\frac{\alpha_{0,1}\alpha_{1,1}}{2\alpha_{0,2}}\Big)\Big].
\end{align*}

We are interested in the leading order behaviour with respect to $\eps$.  From (\ref{alphas}) and using $n_0 = \mathcal{O}(\eps^{-1})$ we see that $\alpha_{2,0}, \alpha_{1,1}, \alpha_{0,2}$ are all $\mathcal{O}(1)$ whilst $\alpha_{1,0}$ and $\alpha_{0,1}$ are $\mathcal{O}(\eps^{-1/2})$.  Hence 
$
	\tfrac{2\alpha_{0,2}}{\alpha_{1,1}^2-4\alpha_{2,0}\alpha_{0,2}} = \mathcal{O}(1)$, and 
	$\alpha_{1,0}-\tfrac{\alpha_{0,1}\alpha_{1,1}}{2\alpha_{0,2}} = \mathcal{O}(\eps^{-1/2}),
$
which in particular shows that the prefactor is $\mathcal{O}(\eps^n)$ whist the argument of the Hermite polynomial is $\mathcal{O}(\eps^{-1})$.  Thus, to leading order, only the highest power of the Hermite polynomial contributes, giving
\begin{align*}
	\epsft{\psi_n^-}(k,t)  \approx  & \e{-\frac{\ii}{\eps} t H^-}  \frac{1}{2\sqrt{4 \alpha_{2,0}\alpha_{0,2} - \alpha_{1,1}^2}} \exp\Big[ \frac{ \alpha_{2,0}\alpha_{0,1}^2 + \alpha_{0,2}\alpha_{1,0}^2 - \alpha_{1,0}\alpha_{0,1}\alpha_{1,1} } {\alpha_{1,1}^2 - 4\alpha_{2,0}\alpha_{0,2}}\Big]\\
	& \times (\eta^*+k) \e{-\tfrac{c}{\eps}(\eta^*-p_0)^2}\e{-\tfrac{\tau_c}{2\delta\epsilon}|k-\eta^*|} \e{-\ii \tfrac{\tau_r}{2\delta\epsilon}(k-\eta^*)} \eta^{*p} \chi_{k^2 > 4\delta},
\end{align*}
which is precisely (\ref{explicitForm}) with the Gaussian replaced by $\eta^p \exp\big({-c}(\eta-p_0)^2/\eps\big)$.

We note that the error in this closed form is expected to be of order $\sqrt{\eps}$.  Whilst this could be improved by taking further terms in the expansion, in the following we choose to concentrate on the case of a wave packet which has been decomposed into a linear combination of complex Gaussians.  The main reason for this is the heuristic phase correction which is discussed in Section \ref{S:PhaseShift}.  From numerical studies, we see that this works well only for Gaussian wave packets, and without this correction, the relative error between the formula and the `exact' numerical wave packet is of the order of $10\%$, compared to an error of around $2\%$ for a Gaussian wave packet with the phase correction.

\subsection{General wave packets as superpositions of Gaussians} \label{S:GeneralWfn}
Due to the strong reliance on the wave packet being a Gaussian in the preceding discussion, formula (\ref{explicitForm}) is not immediately 
applicable to general wave packets.  However, we propose a simple algorithm which allows use of (\ref{explicitForm}).  For a given semi-classical 
wave packet specified on the upper level well away from the crossing region, we evolve it using the BO dynamics on the upper level until the mean 
position of the wave packet coincides with the crossing point (which we choose without loss of generality to be at $x=0$).  We then transform 
into Fourier space and decompose into complex Gaussians, giving a wave packet of the form
\be
	\sum_{j=1}^N A_j \exp\big( -\tfrac{(p-p_i)^2}{\sigma_j^2 \eps} \big) \exp\big(\ii \tfrac{p x_j}{\eps} \big),
	\label{linearcombination}
\ee
where in position space $x_j$ is the offset from the crossing point.

We now need to deal with the fact that the Gaussians reach the crossing point at different times.  We fist note that, for small $\eps$, this 
should be a small effect for semiclassical wavepackets:  since the wave packet is localised in a $\sqrt{\eps}$ neighbourhood of zero in position 
space, we have $x_j = \mathcal{O}(\sqrt{\eps})$.  As discussed in Section \ref{S:ApproxProp}, on a $\sqrt{\eps}$ neighbourhood of the origin and 
for times of order $\sqrt{\eps}$, the dynamics are well-approximated by the explicit propagators (\ref{propagators}).  Since the wave packets 
move with speed of order one, this is still the region of interest and we may simply insert the complex Gaussian into (\ref{full time evol}).  

Applying the rescaling as described in Section \ref{S:Rescaling} gives an extra term in the exponent in (\ref{rescaledg}) of the form $\ii x_j 
(\eta^* + \eps^{1/2} \eta)/\eps$.  The $\eta^*$ term combines with the Gaussian term in $\eta^*$ to give the wave packet evaluated at $\eta^*$ as 
before.  The remaining term provides a contribution of the form $\ii x_j / \eps^{1/2}$ to $\alpha_{1,0}$ in (\ref{alphas}).  

It is now easy to see that in the small $\eps$ limit this term is negligible.  Since $x_j = \mathcal{O}(\sqrt{\eps})$ we see that the new term in 
$\alpha_{1,0}$ is order one.  In contrast, the dominant terms in $\alpha_{1,0}$ are of order $\eps^{-1/2}$ and one may apply (\ref{explicitForm}) 
directly to the complex Gaussian.

We note that for the values of $\eps$ under consideration in the numerics, ignoring this correction increases the relative error by the order of 
$0.1\%$, which is quite significant given the high accuracy of the final formula. In the implementation of the non-Gaussian wave packet below,  
we therefore included the additional term  $\ii x_j / \eps^{1/2}$ in the expression for $\alpha_{0,1}$.  

The above analysis suggests a simple and efficient algorithm for calculating the form of the transmitted wave packet, even if not of Gaussian
form, given an initial wave 
packet $\psi_{-\infty}$ located well away from the transition point in position space.
\begin{enumerate}
\item Evolve the initial wave packet on the upper BO level using the uncoupled BO dynamics until its centre of mass reaches the transition point.  
This can either be pre-determined by finding the point at which the two energy levels are closest, or, as would be required in higher dimensions, 
by monitoring the energy gap at the centre of mass over time and determine its minimum.
\item Transform the resulting wave packet into momentum space and decompose into a linear combination of complex Gaussians as in (\ref
{linearcombination}). 
\item Apply formula (\ref{numericsFormula}) to each complex Gaussian in turn and take the corresponding linear combination.
\item Evolve the resulting transmitted wave packet using the BO dynamics on the lower level, until the centre of mass reaches the scattering 
region.
\end{enumerate}

Assuming that the energy levels become constant in the scattering regime, the computed wave packet will agree up to 
small errors with that computed using the full coupled dynamics.

Note that step (2) may be accomplished using a standard numerical recipe such as non-linear least squares optimization. In practice the formula is more accurate for narrow Gaussians (since this improves a number of the approximations including the choice of fixed $n_0$ and the heuristic phase shift) and thus it may be worth constraining the variances of the Gaussians.  Since the application of the formula is cheap (simply multiplications in Fourier space over a region in which the modulus of the wave packet is significant -- comparable to one time step in uncoupled B-O dynamics) and step (3) scales linearly with the number of Gaussians, increasing the number of Gaussians whilst decreasing their variances would be a reasonable approach to increase accuracy.

It is important to realise that, although this algorithm performs a molecular dynamics calculation using Gaussian wave 
packets, it does not share the obstructions of most Gaussian-based methods (see e.g.\ \cite{Saw85}). These occur mainly due 
to the Gaussians being not orthogonal, and the resulting ill-conditioning of various matrices under time evolution.  
Since we only require that the wave packet is decomposed into Gaussians at the crossing point, transmitted, and re-summed 
on the lower level, we do not encounter such problems.  In fact, one is free to choose any method of propagation on the 
adiabatic levels, for example the method of Hagedorn wave packets \cite{Lubich08}, something which will be important in 
higher dimensions, where simple grid-based methods are prohibitively expensive.

\section{Numerics} \label{S:Numerics}

\begin{figure}
	\includegraphics[width=0.9\linewidth]{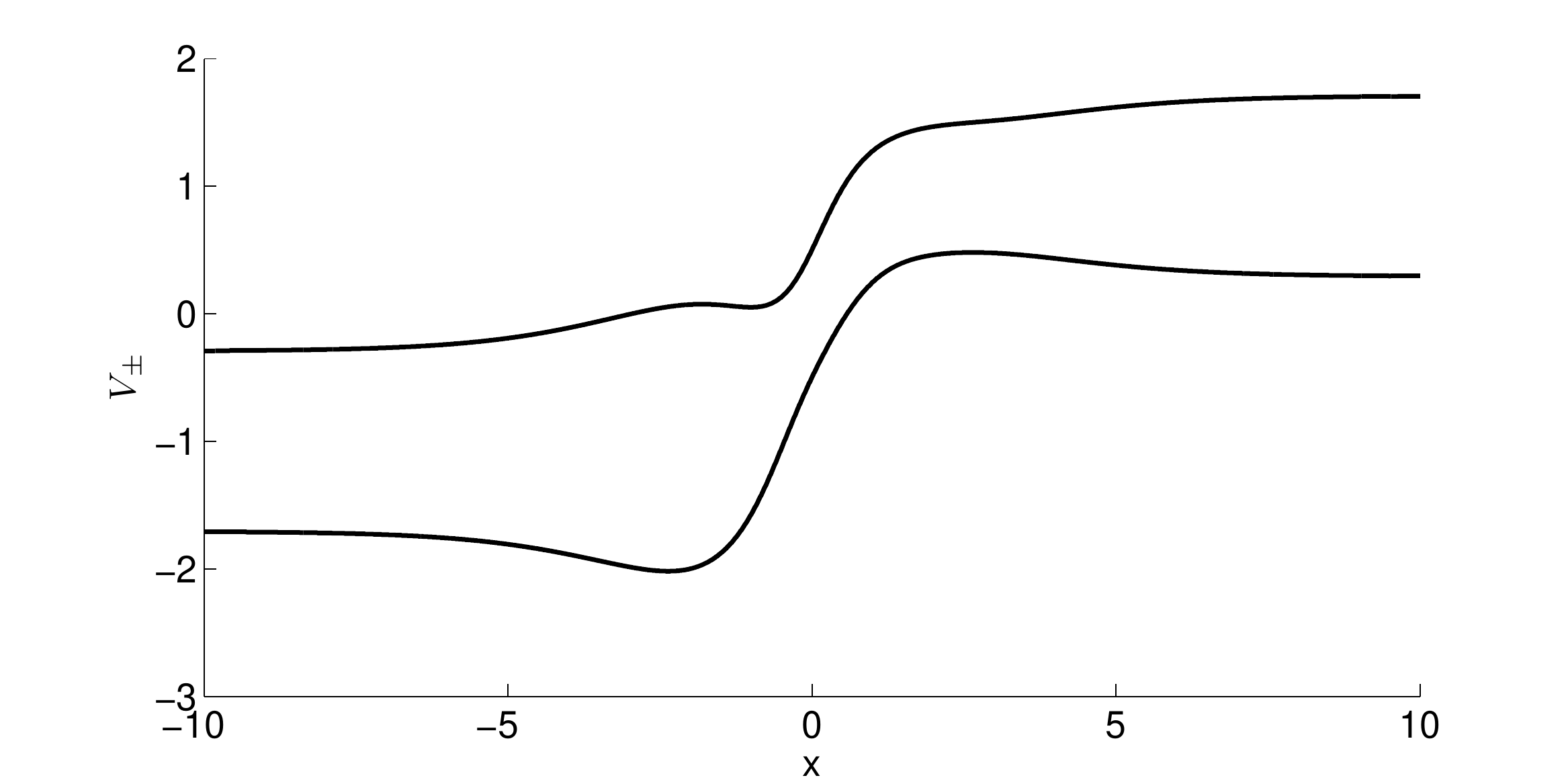}
	\caption{The two adiabatic energy surfaces $V_{\pm}=\pm\rho(x) + d(x)$ with $\rho(x)=\sqrt{X^2(x)+Z^2(x)}$. $Z=\alpha \tanh(x) + \beta x^2/\cosh(x)$, $X=\delta$, $d=\lambda \tanh(x)$, with the parameters $\alpha=0.5$, $\beta=-0.4$, $\delta=0.5$ and $\lambda=1$.  Note that the avoided crossing (minimum of the energy gap) is at $x=0$.}
	 \label{F:Potentials}
\end{figure}

\begin{figure}
	\includegraphics[width=0.9\linewidth]{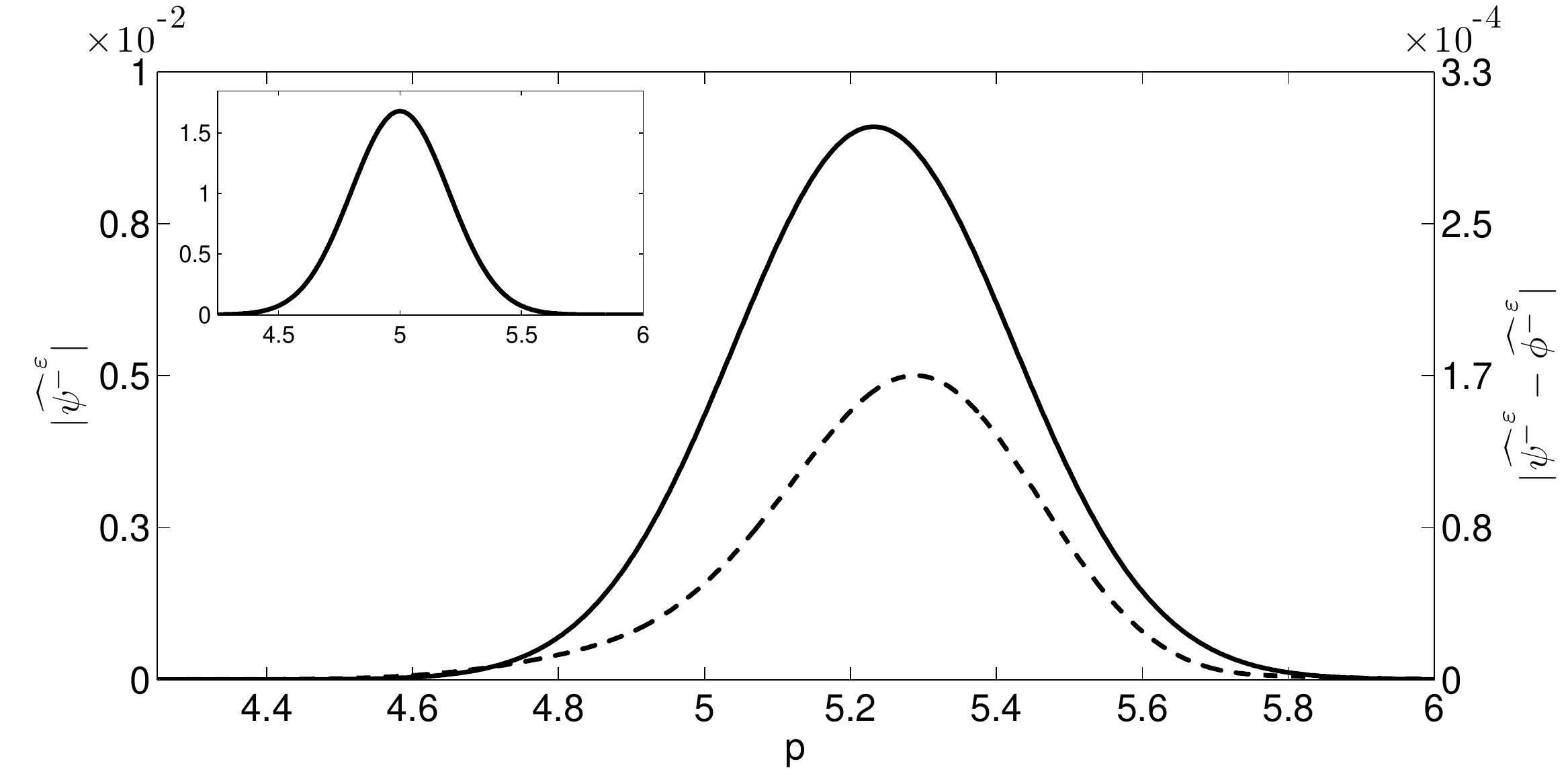}\\
	\includegraphics[width=0.9\linewidth]{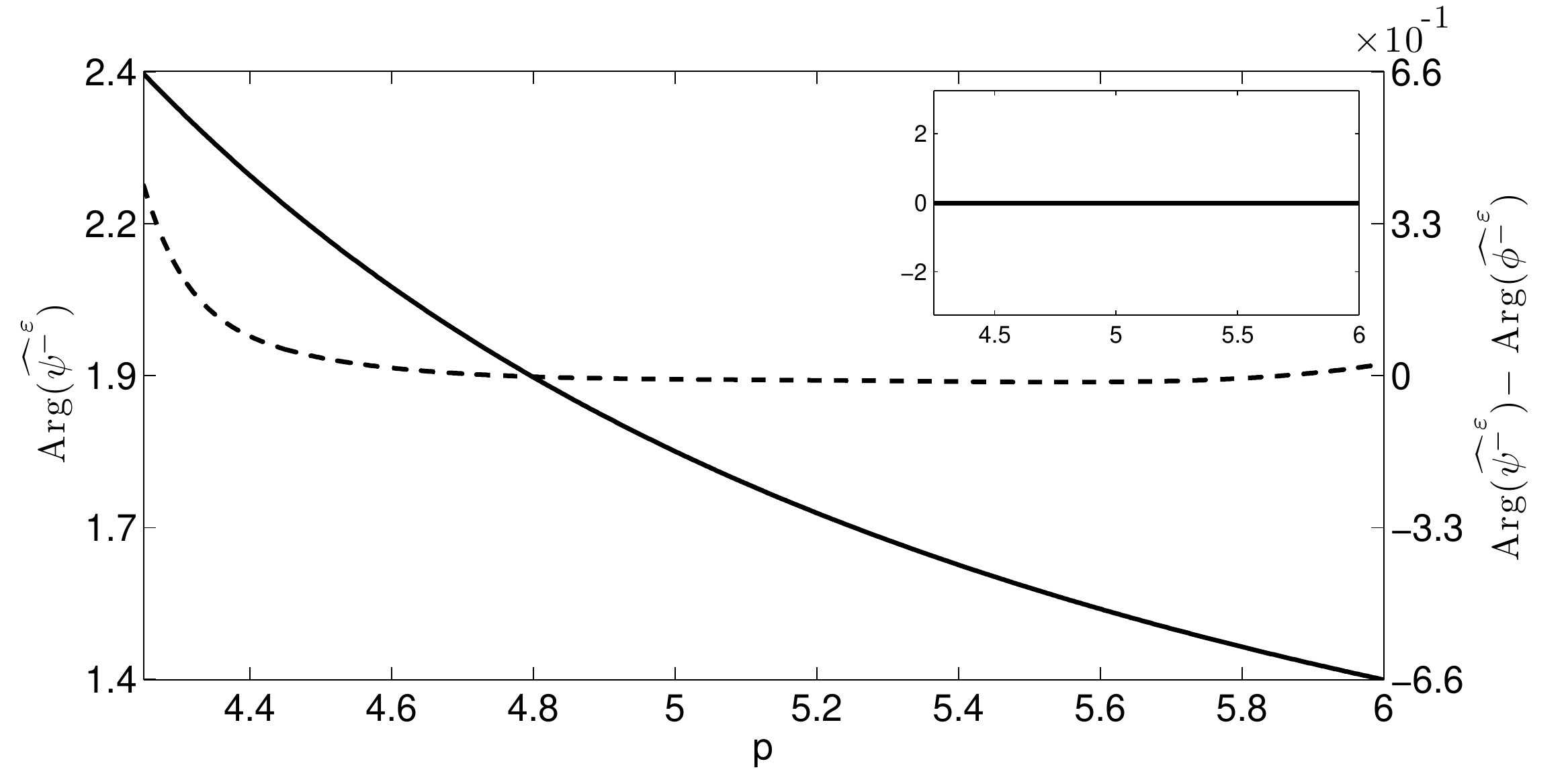}
	\caption{Top: The absolute value of the transmitted wave packet $\epsft{\psi^-}$ as given by (\ref{numericsFormula}) (solid, left axis) and the error as compared to the numerical solution $\epsft{\phi^-}$ computed as described using the fully coupled dynamics (dashed, right axis).  Inset is the initial Gaussian wave packet in momentum space at the transition time.
	Bottom: As in the top plot but showing the argument (phase) of the wave packets.}
	\label{F:Gaussian}
\end{figure}

\begin{figure}
	\includegraphics[width=0.9\linewidth]{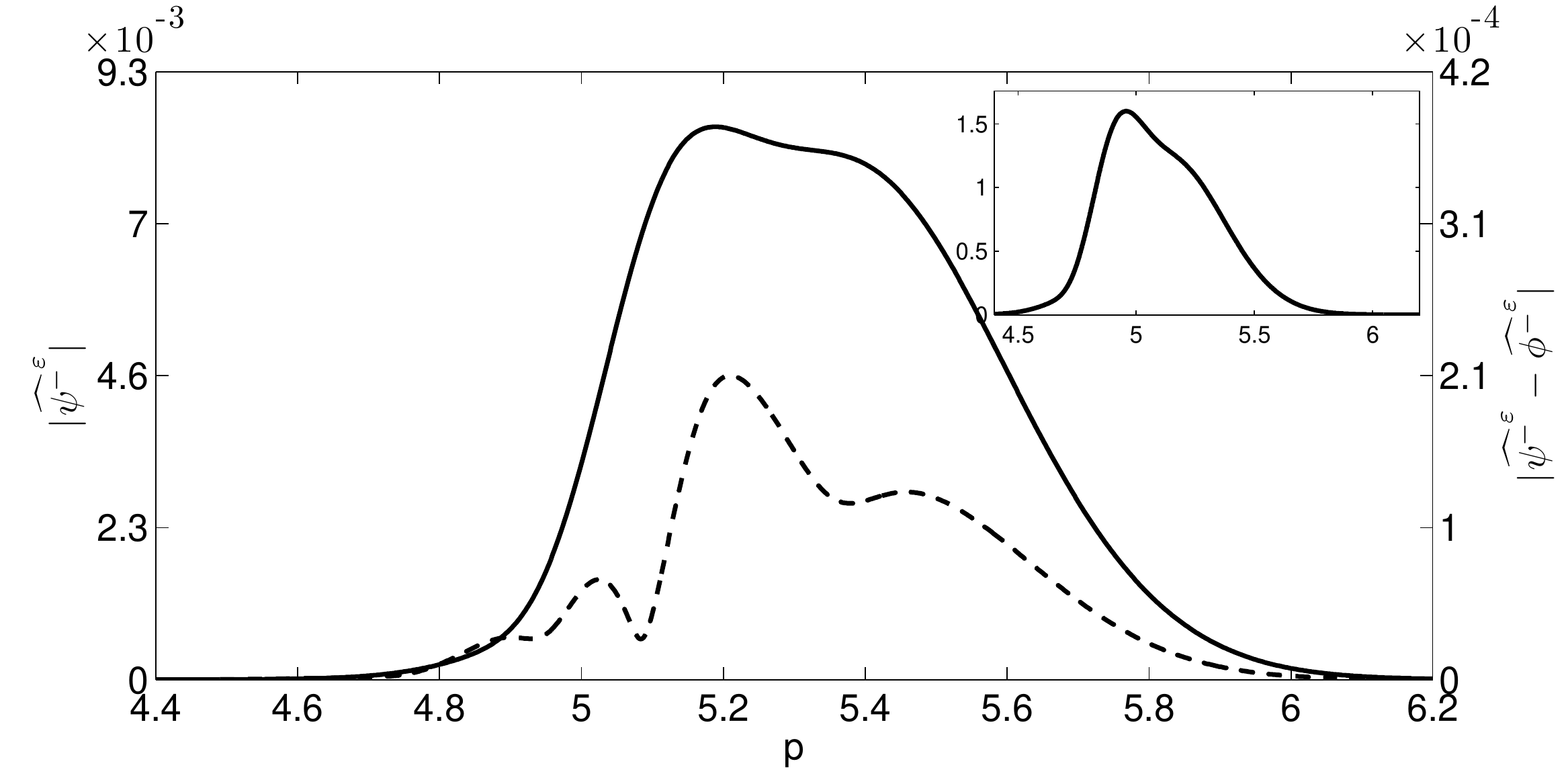}\\
	\includegraphics[width=0.9\linewidth]{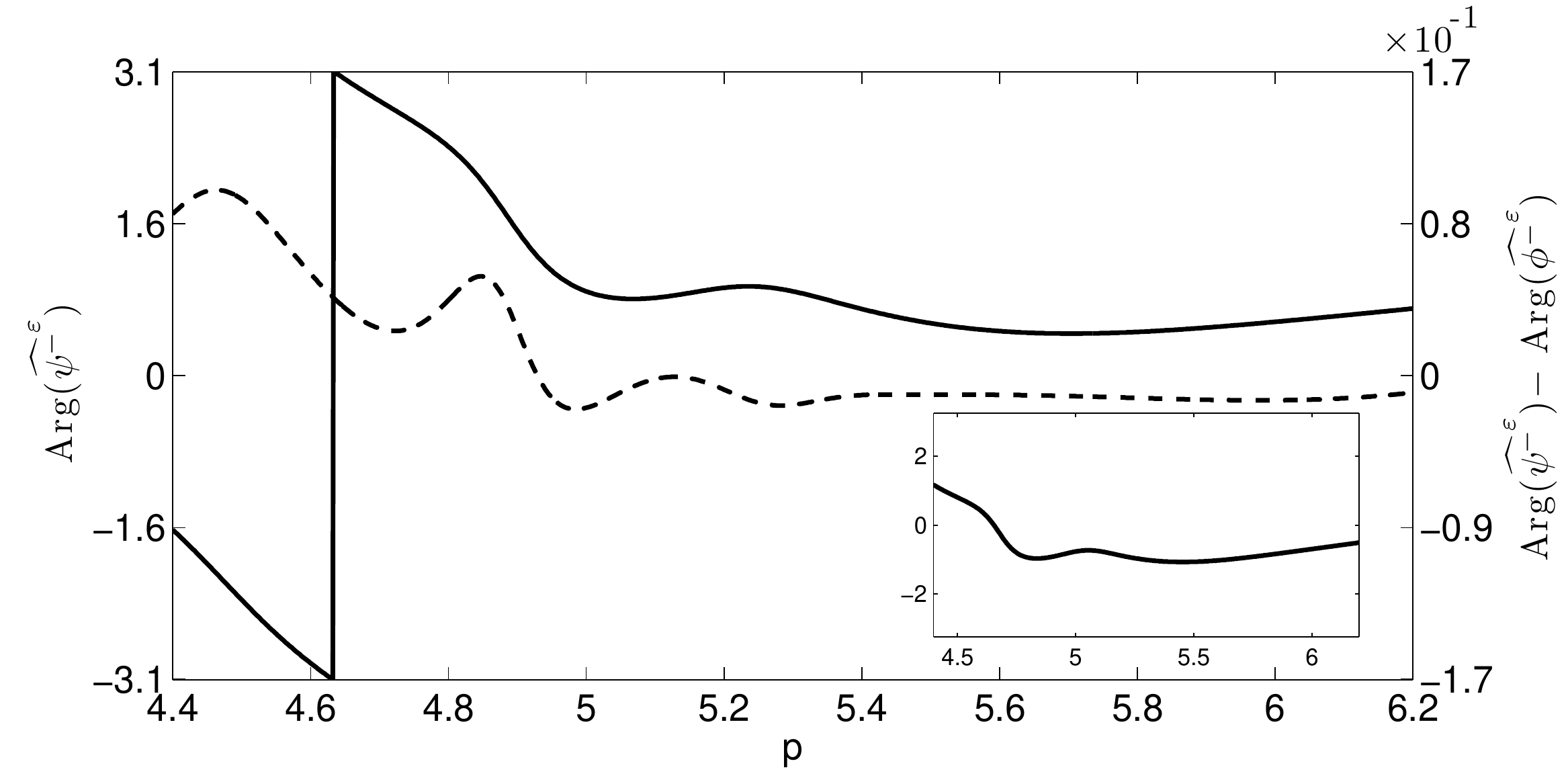}
	\caption{As Figure \ref{F:Gaussian} but for the non-Gaussian wave packet described in the text.}
	\label{F:NonGaussian}
\end{figure}

We now compare the results of formula (\ref{numericsFormula}) to those of high-precision fully-coupled numerics.  
For ease of demonstration, we set the transition point to be at $x=0$ and $t=0$ and choose to specify the initial wave packet $\phi$ as a linear 
combination of complex Gaussians in momentum space at the crossing time.  This simplifies the implementation of the above algorithm, as $\epsft{\phi}$ is already in the 
required form.  Further, we set $\eps=1/50$, which gives reasonably small transition probabilities, whilst still enabling the `exact' calculations to be performed.  Note that, when transformed to position 
space, both examples have mean position zero.  

To begin both the full numerics and the implementation of the above algorithm, we evolve $\phi$ on the upper BO surface to large negative time 
(i.e. to a position where the potentials are essentially flat) to give a good approximation $\phi_{-T} \approx \psi_{-\infty}$.

The full numerics were performed using a symmetric Strang splitting in {\sc matlab} with initial condition $\phi_{-T}$, which is run to a time 
$t_*>0$ where once again the potentials are essentially flat.  In particular, for times  $t > t_*$, the lower component $\|\psi^-_0(t)\|$ is 
constant.  We then evolve $\psi^-_0(t_*)$ backwards in time to $t=0$ and compare its Fourier transform to formula (\ref{numericsFormula}).  The 
calculation was performed on a grid with 16,384 points in both the position ($[-40,40]$) and corresponding momentum ($[-12.87,12.87]$) spaces, 
with $T=t_*=4$ and 1000 time steps.  Doubling both the  number of space and time gridpoints produces a wave function which differs from this computation by around $0.01\%$  in the $L^2$ norm, and hence we take the numerical simulation to be `exact'.

We choose $Z=\alpha \tanh(x) + \beta x^2/\cosh(x)$, $X=\delta$ and $d=\lambda \tanh(x)$.  For these choices, $\delta$ and $\lambda$ correspond to 
their earlier use, the ratio $\alpha^2/\delta$ determines the second derivative of $\rho$ at the transition point, and $\beta$ primarily affects 
the asymmetry of the potential.  In particular, $\beta=0$ gives $\tau_r=0$.  We set $\alpha=0.5$, $\beta=-0.4$, $\delta=0.5$ and $\lambda=1$.  
This leads to the two potential surfaces given in Figure \ref{F:Potentials}, with $\tau_\delta=-0.16611+0.53772\ii$, which can be easily 
calculated numerically.

The first wave packet we treat is given by the complex Gaussian $A \exp \big( {-c} (p-p_0)^2 /(2\eps) \big)$, with $p_0$=5, $c=1/(2\sigma^2)$, $\sigma=\sqrt{2}$ and $A$ chosen such that 
the wave packet is normalized in $L^2$.  The second case we consider is a linear combination of three complex Gaussians of the form (\ref
{linearcombination}) where $|A_j|=A$, $j=1,2,3$, which in turn is chosen to normalize the wave packet.  The remaining parameters are given, with 
$c=1/(2\sigma^2)$ by
\begin{center}
\begin{tabular}{cccc}
	\hline
	$A_j$ & $p_j$ & $\sigma_j$ & $x_j$ \\
	\hline
	$A$ & 5.00 & 1.414 & -0.0238\\
	$A$ & 5.15 & 1.664 & 0.0186\\
	$-A$ & 4.90 & 0.714 & 0.0328\\
	\hline
\end{tabular}
\end{center}

In both cases, the relative error is less than $2\%$ over the full interval where the transmitted wave function is essentially supported. 
The transition probability $\| \psi^- \|^2$ in both cases is of the order $10^{-5}$ ($3.03 \times 10^{-5}$ and $3.48 \times 10^{-5}$ for the 
Gaussian and non-Gaussian cases respectively).  In addition to these two examples, we have tested a wide range of parameters for the  both the 
potentials and semi-classical wave functions, and all results are good to within a few percent.  They deteriorate only when $\eps$ (and thus also 
$\|\psi^-\|$) becomes too large and we leave the adiabatic regime, or when $p_0$ (and thus also $\|\psi^-\|$) gets too small and our many 
approximations requiring that $p_0$ is suitably large break down.  In particular, the relative error is less than a few percent when the 
transition probability is in the range $10^{-2}$--$10^{-15}$.


\begin{thebibliography}{99}

\bibitem{BG}
V.~Betz and B.~D. Goddard.
\newblock Accurate prediction of non-adiabatic transitions through avoided
  crossings.
\newblock {\em Phys. Rev. Lett.}, 103:213001, 2009.

\bibitem{BG2}
V.~Betz and B.~D. Goddard.
\newblock Transitions through avoided crossings in the high momentum regime.
\newblock {\em in preparation}, 2010.

\bibitem{BGT}
V.~Betz, B.~D. Goddard, and S.~Teufel.
\newblock Superadiabatic transitions in quantum molecular dynamics.
\newblock {\em Proc. Roy. Soc. A}, 465(2111):3553--3580, 2009.

\bibitem{BeLi93}
MV~Berry and R~Lim.
\newblock Universal transition prefactors derived by superadiabatic
  renormalization.
\newblock {\em J Phys A-Math Gen}, 26(18):4737--4747, Jan 1993.

\bibitem{BT05-2}
Volker Betz and Stefan Teufel.
\newblock Precise coupling terms in adiabatic quantum evolution: the generic
  case.
\newblock {\em Comm. Math. Phys.}, 260(2):481--509, 2005.

\bibitem{Ha94}
George~A. Hagedorn.
\newblock Molecular propagation through electron energy level crossings.
\newblock {\em Mem. Amer. Math. Soc.}, 111(536):vi+130, 1994.

\bibitem{HaJo04}
George~A Hagedorn and Alain Joye.
\newblock Time development of exponentially small non-adiabatic transitions.
\newblock {\em Comm. Math. Phys.}, 250(2):393--413, 2004.

\bibitem{HaJo05}
George~A Hagedorn and Alain Joye.
\newblock Determination of non-adiabatic scattering wave functions in a
  born-oppenheimer model.
\newblock {\em Ann. Henri Poincar{\'e}}, 6(5):937--990, 2005.

\bibitem{Las08}
C.~Lasser and T.~Swart.
\newblock Single switch surface hopping for a model of pyrazine.
\newblock {\em J. Chem. Phys}, 129:034302, 2008.

\bibitem{Las07}
C.~Lasser, T.~Swart, and S.~Teufel.
\newblock Construction and validation of a rigorous surface hopping algorithm
  for conical crossings.
\newblock {\em Commun. Math. Phys.}, 5:789--814, 2007.

\bibitem{Lubich08}
C.~Lubich.
\newblock {\em From Quantum to Classical Molecular Dynamics: Reduced Models and
  Numerical Analysis}.
\newblock European Math. Soc., 2008.

\bibitem{Nak02}
H.~Nakamura.
\newblock {\em Nonadiabatic Transition}.
\newblock World Scientific, Singapore, 2002.

\bibitem{RRZ89}
Todd~S. Rose, Mark~J. Rosker, and Ahmed~H. Zewail.
\newblock Femtosecond real-time probing of reactions. iv. the reactions of
  alkali halides.
\newblock {\em The Journal of Chemical Physics}, 91(12):7415--7436, 1989.

\bibitem{Saw85}
S-I Sawada, R.~Heather, B.~Jackson, and H.~Metiu.
\newblock A strategy for time dependent quantum mechanical calculations using a
  gaussian wave packet representation of the wave function.
\newblock {\em J. Chem. Phys}, 83:3009--3027, 1985.

\bibitem{Teu03}
Stefan Teufel.
\newblock {\em Adiabatic perturbation theory in quantum dynamics}, volume 1821
  of {\em Lecture Notes in Mathematics}.
\newblock Springer-Verlag, Berlin, 2003.

\bibitem{Tul90}
J.~C. Tully.
\newblock Molecular dynamics with electronic transitions.
\newblock {\em J. Chem. Phys}, 93:1061--1071, 1990.

\bibitem{Vor98}
A.~I. Voronin, J.~M.~C. Marques, and A.~J.~C. Varandas.
\newblock Trajectory surface hopping study of the li +
  li$_2$(x$^1$$\sigma_g^+$) dissociation reaction.
\newblock {\em J. Phys. Chem. A}, 102(30):6057--6062, 1998.

\bibitem{vNW29}
J.~Von~Neumann and E~Wigner.
\newblock {\"U}ber das {V}erhalten von {E}igenwerten bei {A}diabatischen
  {P}rozessen.
\newblock {\em Phys Z}, 30:467, 1929.

\bibitem{Scho94}
Q.~Wang, R.~W. Schoenlein, L.~A. Peteanu, R.~A. Mathies, and C.~V. Shank.
\newblock Vibrationally coherent photochemistry in the femtosecond primary
  event of vision.
\newblock {\em Science}, 266(5184):422--424, 1994.

\bibitem{Zen32}
D.~Zener.
\newblock Non-adiabatic crossings of energy levels.
\newblock {\em Proc. Roy. Soc. London}, 137:696--702, 1932.

\end{thebibliography}
\end{document}